\documentclass[sn-mathphys-num]{sn-jnl}

\usepackage{amssymb}
\usepackage{amsthm}
\usepackage{graphicx}
\usepackage{amsmath}
\usepackage{amsfonts}
\usepackage{color}
\usepackage{centernot}
\usepackage{mathtools}
\usepackage{stmaryrd}

\theoremstyle{thmstyleone}%
\newtheorem{theorem}{Theorem}

\newcommand{\newtext}[1]{{#1}}

\begin{document}

\title{RG analysis of spontaneous stochasticity on a fractal lattice: stability and bifurcations}
\author{Alexei A. Mailybaev} 
\affil{Instituto de Matem\'atica Pura e Aplicada -- IMPA, Rio de Janeiro, Brazil. \\ Email: alexei@impa.br}

\abstract{
\newtext{In this paper, we study the stability and bifurcations of spontaneous stochasticity using an approach reminiscent of the Feigenbaum renormalization group (RG).}
We consider dynamical models on a self-similar space-time lattice as toy models for multiscale motion in hydrodynamic turbulence.
Here an ill-posed ideal system is regularized at small scales and the vanishing regularization (inviscid) limit is considered. 
By relating the inviscid limit to the dynamics of the RG operator acting on the flow maps, we explain the existence and universality (regularization independence) of the limiting solutions as a consequence of the fixed-point RG attractor.
Considering the local linearized dynamics, we show that the convergence to the inviscid limit is governed by the universal RG eigenmode. 
We also demonstrate that the RG attractor undergoes a period-doubling bifurcation with parameter variation, thereby changing the nature of the inviscid limit. 
In the case of chaotic RG dynamics, we introduce the stochastic RG operator acting on Markov kernels. 
Then the RG attractor becomes stochastic, which explains the existence and universality of spontaneously stochastic solutions in the limit of vanishing noise. 
We study a linearized structure (RG eigenmode) of the stochastic RG attractor and its period-doubling bifurcation.
Viewed as prototypes of Eulerian spontaneous stochasticity, our models explain its mechanism, universality and potential diversity.
}


\maketitle

\section{Introduction} 

Developed hydrodynamic turbulence is characterized by very large Reynolds numbers, which is mathematically equivalent to the inviscid limit of the incompressible Navier-Stokes equations~\cite{frisch1999turbulence}. 
It has long been predicted that turbulent solutions are extremely sensitive to tiny small-scale uncertainties such as microscopic noise~\cite{lorenz1969predictability,leith1972predictability,ruelle1979microscopic,eyink1996turbulence,palmer2014real,boffetta2017chaos}. 
As a consequence, the finite-time forecasting problem cannot be formulated without adding small noise terms to the Navier-Stokes equations.

With the aim of understanding this finite-time unpredictability, we focus on the initial value problem and its properties in the inviscid limit.
Recent studies of the Navier-Stokes system~\cite{thalabard2020butterfly} and its simplified (shell) models~\cite{mailybaev2016spontaneously,mailybaev2016spontaneous,mailybaev2017toward,biferale2018rayleigh,bandak2024spontaneous} have shown that solutions of the initial value problem converge in the limit where both viscosity and noise tend to zero.
However, this limit remains stochastic.
It represents a stochastic process solving a deterministically formulated initial value problem for the ideal (incompressible Euler) equations~\cite{eyink2024space}. 
This dichotomy of velocity fields that are stochastic and that solve a formally deterministic problem is called Eulerian spontaneous stochasticity.
Here the word "Eulerian" is used to distinguish it from the similar phenomenon of Lagrangian spontaneous stochasticity, which describes the stochasticity of trajectories in rough (H\"older continuous) vector fields~\cite{bernard1998slow,falkovich2001particles,kupiainen2003nondeterministic}. 
It was also noted that Eulerian spontaneous stochasticity is universal: the limiting stochastic process does not depend on the form of the (vanishing) dissipative and noise terms. 

This paper proposes a theory that describes the mechanism of spontaneous stochasticity, its universality and potential diversity.
For this purpose, we use a class of models that are spontaneously stochastic and convenient for both theoretical and numerical study.
As in shell models~\cite{biferale2003shell} we represent physical space as a discrete sequence of scales.
Using Poincar\'e's idea of replacing continuous time with discrete time, we evolve each scale of our model by its own turnover time.
Then the dynamics is reduced to functional relations on a fractal space-time lattice. 
This type of models was introduced in~\cite{mailybaev2023spontaneously}, where a rigorous example of a universal spontaneously stochastic system was presented. 
Here we consider a slight modification that takes into account energy balance. 

The inviscid limit on a fractal lattice is represented by a sequence of well-posed regularized systems converging to a self-similar and ill-posed ideal system. Following~\cite{mailybaev2023spontaneous} we define a renormalization group (RG) operator acting on the flow maps of regularized systems. 
Then the inviscid limit is represented as an attractor of the RG dynamics.
This approach has similarities with Feigenbaum's RG theory~\cite{feigenbaum1983universal}: reducing the viscous scale amounts to a double action (composition) of flow maps combined with the proper rescaling. \newtext{In this sense, our approach is different from RG analysis via small-scale coarse graining in the Kadanoff--Wilson spirit~\cite{wilson1983renormalization}; for the applications to turbulence we refer (among many others) to the early works~\cite{forster1977large,yakhot1986renormalization} and the non-perturbative theory~\cite{canet2022functional}. Our RG approach upgrades the dynamics at the largest scale, thereby, correlating with the ideas that the renormalization in turbulence can be implemented in the inverse fashion, from the side of large scales~\cite{eyink1994analogies,inverseRG}.}

In this paper, we use numerical simulations to test the dynamical properties of the RG operator, which are predicted using the classical theory of dynamical systems. First, we verify that the fixed-point RG attractor has a regular local structure governed by the leading eigenmode of the linearized operator. This implies that not only the inviscid limit but also the first-order corrections are universal, because these corrections are determined by the RG eigenvalue and eigenvector. We show that changing the parameter in the ideal model leads to a classical period-doubling bifurcation of the RG attractor when the eigenvalue passes the value $\rho = -1$. This leads to a qualitative change in the dynamics achieved in the inviscid limit. 

Spontaneous stochasticity appears when the RG dynamics is chaotic. In this case, we extend our approach by introducing a stochastic RG operator acting on Markov kernels for regularized systems with noise~\cite{mailybaev2023spontaneous}. Then the spontaneously stochastic limit corresponds to a fixed-point attractor of this stochastic RG operator. 
As a consequence, spontaneously stochastic solutions have properties similar to deterministic ones: the universality of spontaneously stochastic solutions follows from the existence of the RG attractor, and the universality of first-order corrections to their probability distributions follows from the leading RG eigenmode.
Finally, we demonstrate a period-doubling bifurcation of the spontaneously stochastic attractor. 
This bifurcation shows that spontaneous stochasticity can take different forms, and these forms are described by the developed RG theory.
In a parallel work~\cite{mailybaev2024rg}, we extended our RG approach to shell models of turbulence, demonstrating similar results in continuous-time systems. 

\newtext{We remark that spontaneous stochasticity appears in systems without uniqueness, though the latter property alone is not sufficient and does not explain the universality. Nonuniqueness of weak solutions~\cite{buckmaster2019nonuniqueness} fits naturally into the stochastic approach for hydrodynamic turbulence~\cite{vishik2012mathematical}, and justifies the use of Markov attractors~\cite{flandoli2008markov}. The stochastic RG dynamics presented in this paper emerges as a further development of such ideas. Here we do not pretend to give a mathematically rigorous theory but rather develop a consistent formalism that can be verified numerically. Still, due to simplicity of our model, some RG properties are rigorous and formulated as theorems.}

The paper is organized as follows. Section~\ref{sec2} introduces the model. Section~\ref{sec3} describes the RG approach. Section~\ref{sec4} studies the fixed-point RG attractor. Section~\ref{sec5} demonstrates the period-doubling bifurcation of the RG attractor. Section~\ref{sec6} develops the stochastic RG approach. Section~\ref{sec7} relates spontaneously stochastic solutions to the fixed-point attractor of the stochastic RG operator. Section~\ref{sec8} studies the period-doubling bifurcation of spontaneous stochasticity. 
Section~\ref{sec10} summarizes the findings of this work.
Finally, in the Appendix (Section~\ref{sec11}) we briefly discuss the turbulent steady state.

\section{Dynamics on a fractal lattice}
\label{sec2} 

We consider dynamical systems on a fractal space-time lattice shown in Fig.~\ref{fig2}. The spatial dimension of the lattice is represented by a geometric progression of scales $\ell_n = 2^{-n}$ with scale numbers $n = 0,1,2,\ldots$. Each scale has a corresponding turnover time $\tau_n$ to which we attribute the same value $\tau_n = \ell_n = 2^{-n}$. Using the turnover times, we define the discrete time sequence at each scale as $t/\tau_n = 0,1,2,\ldots$.
At each scale and corresponding discrete time we introduce a non-negative state variable $u_n(t) \ge 0$. 
Thus, the overall evolution is determined by the collection of states $\{u_n(t)\}$, where the pairs $(n,t)$ belong to the fractal lattice in Fig.~\ref{fig2}.

\begin{figure}[tp]
\centering
\includegraphics[width=0.9\textwidth]{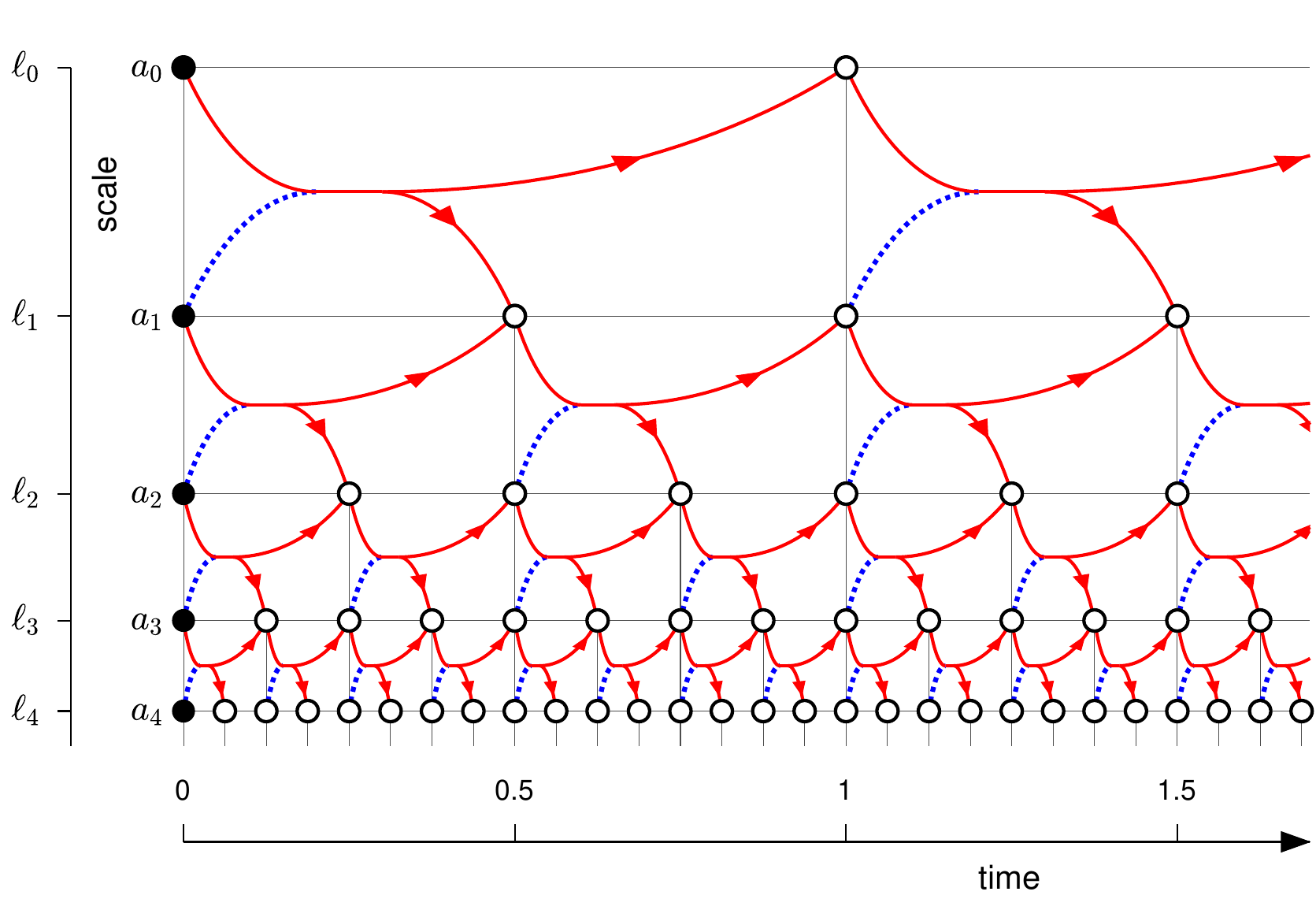}
\caption{Fractal space-time lattice. The initial conditions $u_n(0) = a_n$ are shown as black dots, and the unknown states $u_n(t)$ are shown as white dots. The arrows indicate the functional dependence of states at different times: red arrows correspond to energy transfer, and blue dotted lines indicate the influence of adjacent scales on energy transfer.}
\label{fig2}
\end{figure}

\begin{figure}[tp]
\centering
\includegraphics[width=0.7\textwidth]{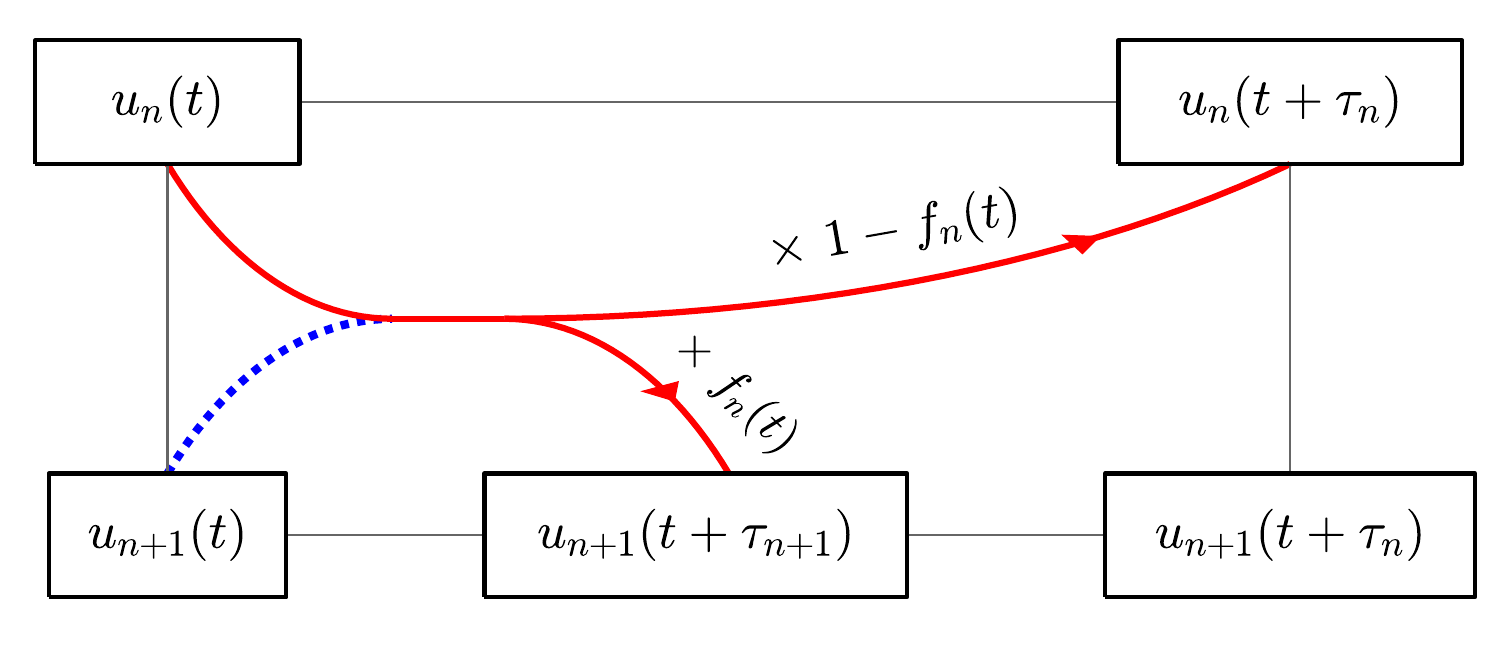}
\caption{Energy transfer scheme: the fraction $f_n(t)$ of $u_n(t)$ is transferred to the smaller scale $n+1$ at the next time $t+\tau_{n+1}$, while the fraction $1-f_n(t)$ remains in the same scale. The value of $f_n(t)$ depends on the states at two adjacent scales, $u_n(t)$ and $u_{n+1}(t)$.} 
\label{fig1}
\end{figure}

The dynamics of the fractal lattice is introduced by defining relations between states at different times. 
We interpret $u_n(t)$ as the energy \newtext{(or any other invariant)} at scale $n$ and consider the following class of conservative interactions.
The fraction $f_n(t)$ of energy $u_n(t)$ is transferred to a smaller scale at each turnover time, while the remaining fraction $1-f_n(t)$ stays at the same scale; see Fig.~\ref{fig1}. The transfer fraction is defined by the same function
	\begin{equation}
	\label{eq2_2b}
	f_n(t) = f\big(u_n(t),u_{n+1}(t)\big)
	\end{equation}
in terms of two adjacent states. Thus, the equations of motion are formulated as
	\begin{equation}
	\label{eq2_2}
	u_n(t+\tau_n) = \big[1-f_n(t)\big] u_n(t) +
	\chi_{\mathrm{even}}\left(\frac{t}{\tau_n}\right) 
	f_{n-1}(t) u_{n-1}(t)
	\end{equation}
at integer turnover times $t/\tau_n = 0,1,2,\ldots$.
The first term on the right-hand side of Eq.~(\ref{eq2_2}) represents the energy remained at scale $n$, and the second term is the energy transferred from scale $n-1$. The function $\chi_{\mathrm{even}}(m) = 1$ for even $m$ and zero otherwise. 

The central property of ideal dynamics (\ref{eq2_2}) is scale invariance: the relations are invariant under the transformation
	\begin{equation}
	\label{eq2_S}
	n \mapsto n+1, \quad t \mapsto t/2,
	\end{equation}
which correspond to the change of scales $\ell_n \mapsto \ell_{n+1}$ and turnover times $\tau_n \mapsto \tau_{n+1}$. This property imitates the space-time scale invariance of ideal hydrodynamics (Euler equations) \newtext{for the specific scaling exponent $h = 0$}~\cite{frisch1999turbulence}. 

The scale invariance is broken at $n = 0$, which corresponds to the largest scale $\ell_0$. Here the governing equations (\ref{eq2_2}) are understood as
	\begin{equation}
	\label{eq2_4}
	u_0(t+1) = \big[1-f_0(t)\big] u_0(t), \quad t = 0,1,\ldots,
	\end{equation}
where we assume that $u_{-1}(t) \equiv 0$; see Fig.~\ref{fig2}. In general, one can add a forcing term to Eq.~(\ref{eq2_4}). 
For simplicity, we do not consider forcing as it is irrelevant to our analysis. 
However, for a brief account of the forced dynamics, see the Appendix.

We consider initial conditions 
	\begin{equation}
	\label{eq2_5}
	u_n(0) = a_n \ge 0, \quad n = 0,1,\ldots
	\end{equation}
with the finite total energy $\sum_{n = 0}^\infty a_n < \infty$. We will discuss in the end of Section~\ref{sec4} that extra conditions on the decay of the sequence $\{a_n\}$ may be required.
Initial conditions (\ref{eq2_5}) with dynamics equations~(\ref{eq2_2}), (\ref{eq2_4}) and (\ref{eq2_2b}) define the ideal initial value problem.
\newtext{In general, its solutions are not unique, which makes the ideal initial value problem ill-posed. Indeed, try to trace the state $u_1(t)$ at time $t = 3\tau_1 = 1.5$ to its predecessors at previous times following red or blue lines in Fig.~\ref{fig2}. By doing so one diverges to infinitesimal scales $n \to \infty$ while staying at a finite distance from the initial time $t = 0$ in infinitely many ways. This gives an idea that $u_1(t)$ cannot be expressed in terms of initial conditions alone. The rigorous proof of nonuniqueness may be complicated, and we refer to \cite{mailybaev2023spontaneously} where the nonuniqueness was studied in detail for a similar model, and to \cite{mailybaev2023spontaneous} for the relation of this nonuniqueness with a finite-time blowup.}

The well-posed formulation is obtained using regularization.
By saying that the system is regularized at a given scale $N$, we mean that the dynamics is modified at the scales $n \ge N$ such that the  solution of the initial value problem exists and is unique. Then the ideal system is recovered in the limit $N \to \infty$. 
We call $N$ the viscous scale and $N \to \infty$ the inviscid limit. 
Our goal is to study how regularized solutions behave in the inviscid limit depending on the choice of regularization.

For our analysis we introduce the following family of regularizations, which depend on two parameters: the viscous scale $N \ge 0$ and a real parameter $0 < \alpha \le 1$.
For fixed $N$ and $\alpha$  the regularized system is defined as follows.
At scales $n \le N$ we keep Eqs.~(\ref{eq2_2}) and  (\ref{eq2_4}) with a modified transfer function
	\begin{equation}
	\label{eq2_7}
	f_n(t) = \left\{ \begin{array}{ll}
	f\big(u_n(t),u_{n+1}(t)\big), & n < N; \\
	\alpha, & n = N.
	\end{array}\right.
	\end{equation}
For scales $n > N$ we impose the trivial dynamics as
	\begin{equation}
	\label{eq2_6}
	u_n(t) = 0, \quad n > N, \ t > 0.
	\end{equation}
The regularized equations coincide with the equations of the ideal system for $n < N$. At the viscous scale $N$, energy dissipates at a constant rate (constant fraction per turnover time) $\alpha$, and then is instantly removed from scales $n > N$.
We regard this regularization as an imitation of the dissipative (viscous) mechanism in turbulence~\cite{frisch1999turbulence}.
At all scales $n > 0$ (away from the large-scale cutoff), our regularized equations have the scaling symmetry 
	\begin{equation}
	\label{eq2_ES}
	n \mapsto n+1, \quad t \mapsto t/2, \quad N \mapsto N+1,
	\end{equation}
which extends the ideal scale invariance (\ref{eq2_S}).  

\section{Inviscid limit as RG dynamics}
\label{sec3}

Let us introduce compact notations $a = (a_0,a_1,\ldots)$ for initial conditions and similarly $u(t) = \big(u_0(t),u_1(t),u_2(t),\ldots\big)$ for the states at integer times $t = 0,1,2,\ldots$.
We assume that both $a$ and $u(t)$ belong to the phase space denoted as $\mathbb{U}$.
For now, $\mathbb{U}$ can be viewed as the space of summable sequences with non-negative components, but additional regularity conditions may be imposed later, as discussed at the end of Section~\ref{sec4}. Recall that the whole field (defined at all points of the lattice) is denoted as $\{u_n(t)\}$.

We define the flow map $\phi^{(N,\alpha)}: \mathbb{U} \mapsto \mathbb{U}$ of the regularized system by the relation
	\begin{equation}
	\label{eq3_1}
	\phi^{(N,\alpha)}(a) = u(1).
	\end{equation}
Here $u(1)$ is the solution of the regularized problem with parameters $(N,\alpha)$ and initial state $a$, evaluated at the large-scale turnover time $\tau_0 = 1$. In particular, for $N = 0$ one obtains 
	\begin{equation}
	\label{eq3_1x}
	\phi^{(0,\alpha)}(a) = \big((1-\alpha) a_0,0,0,\ldots\big).
	\end{equation}

Let us define the shift maps and the large-scale projector
	\begin{equation}
	\label{eq3_2}
	\sigma_+(a) = (a_1,a_2,\ldots), \quad
	\sigma_-(a) = (0,a_0,a_1,\ldots), \quad
	\pi_0(a) = (a_0,0,0,\ldots),
	\end{equation}
and the energy transfer map
	\begin{equation}
	\label{eq3_3}
	\xi(a) = (f(a_0,a_1)a_0,0,0,\ldots).
	\end{equation}
Then, the flow maps for different viscous scales $N$ are explicitly related by the following

\begin{theorem}
\label{th1}
For any $N \ge 0$, the following relation holds
	\begin{equation}
	\label{eq3_4}
	\phi^{(N+1,\alpha)} = \mathcal{R}[\phi^{(N,\alpha)}],
	\end{equation}
where the operator $\mathcal{R}$ acting on maps $\phi: \mathbb{U} \mapsto \mathbb{U}$ is defined as
	\begin{equation}
	\label{eq3_5}
	\mathcal{R} [\phi] = \pi_0-\xi+\sigma_- \circ \phi \circ (\xi+\phi \circ \sigma_+).
	\end{equation}
This operator is called the renormalization group (RG) operator.
\end{theorem}

\begin{proof}
The RG relations (\ref{eq3_4}) and (\ref{eq3_5}) are the consequences of the self-similarity of the ideal system. 
This can be seen in Fig.~\ref{fig3}, where the larger (red) rectangle denotes a regularized system with viscous scale $N+1$. Then one can see that the systems bounded by the two smaller (green) rectangles are identical to the regularized systems with viscous scale $N$. In particular, denoting the intermediate state $u' = (u_1(t'),u_2(t'),\ldots)$ at time $t' = 1/2$, one finds
	\begin{equation}
	\label{eq3_6}
	u' = \xi(a)+\phi^{(N,\alpha)}(a'), \quad a' = \sigma_+(a) = (a_1,a_2,\ldots).
	\end{equation}
Here the second term $\phi^{(N,\alpha)}(a')$ corresponds to the left green rectangle in Fig.~\ref{fig3}, and first term $\xi(a)$ accounts for the energy transfer from the scale $n = 0$ to $1$. Similarly, using the flow map $\phi^{(N,\alpha)}$ for the right green rectangle, one finds
	\begin{equation}
	\label{eq3_7}
	\phi^{(N+1,\alpha)}(a) = \pi_0(a)-\xi(a)+\sigma_- \circ \phi^{(N,\alpha)}(u'),
	\end{equation}
where $\pi_0(a)-\xi(a)$ specifies the energy remained at scale $n = 0$; see Eq.~(\ref{eq2_2}) and Fig.~\ref{fig1}.
Combining the identities (\ref{eq3_6}) and (\ref{eq3_7}) yields the RG relations (\ref{eq3_4}) and (\ref{eq3_5}).
\end{proof}

\begin{figure}[tp]
\centering
\includegraphics[width=0.7\textwidth]{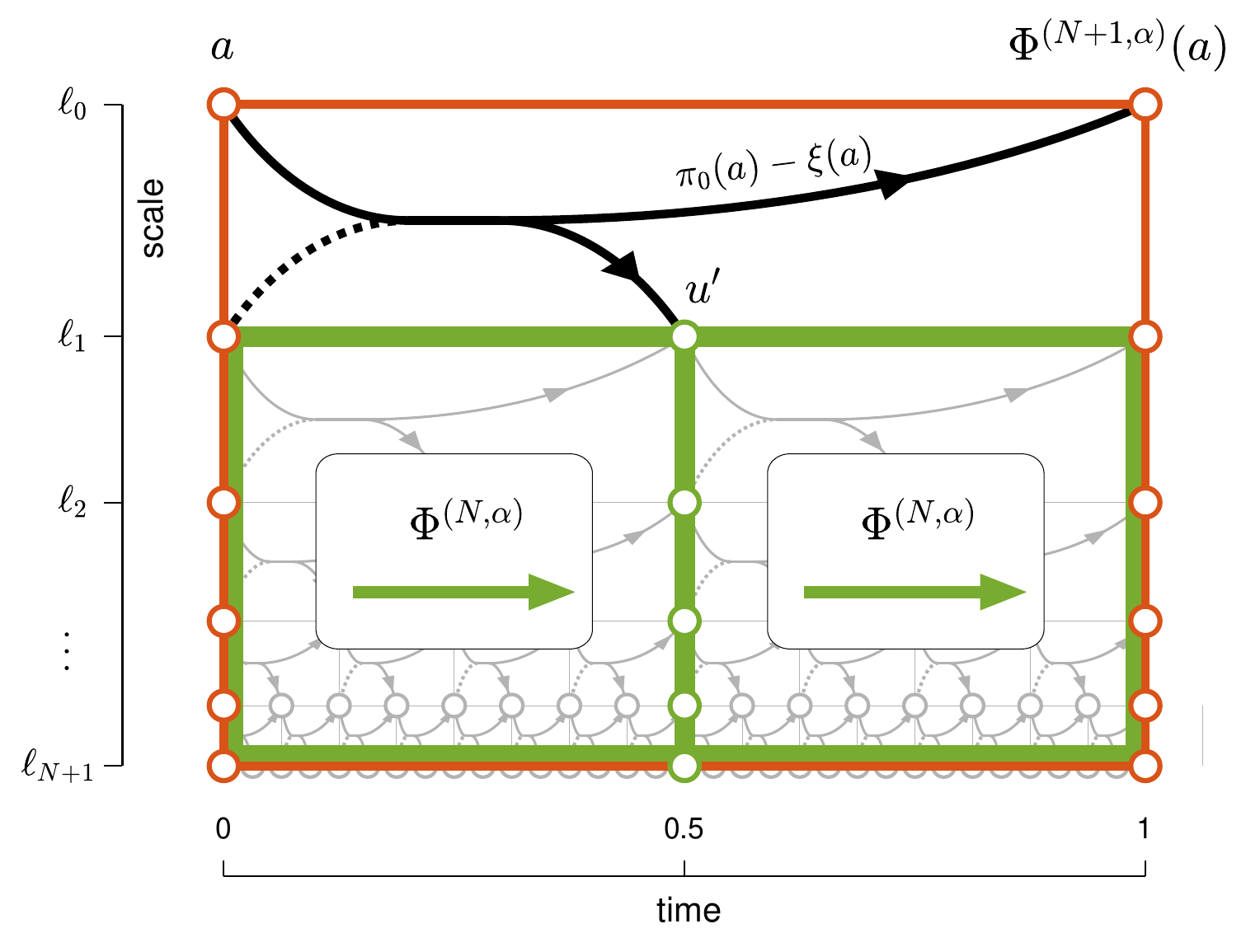}
\caption{Schematic representation of the RG operator transforming $\phi^{(N,\alpha)}$ to $\phi^{(N+1,\alpha)}$.} 
\label{fig3}
\end{figure}

We see that the increase of the viscous scale $N$ is equivalent to the RG dynamics (\ref{eq3_4}) starting from the initial flow map (\ref{eq3_1x}). 
\newtext{These flow maps, however, solve the regularized problems at the particular time $t = 1$ only. We now show that solutions at other times are obtained as a proper composition of flow maps with different $N$. This reduces solving the regularized initial value problem to determining the flow maps.}

\begin{theorem}
\label{th2}
Let $\{u_n(t)\}$ be the solution of the regularized initial value problem for the initial condition $a$ and parameters $(N,\alpha)$. 
Consider time $t = m+\tau_{i_1}+\ldots+\tau_{i_k}$ for some integers $m \ge 0$ and $0 = i_0 < i_1 < \cdots < i_k \le N$.
Then, all solution components at this time take the form
	\begin{equation}
	\label{eq3_GS}
	\big(u_{i_k}(t),u_{i_k+1}(t),u_{i_k+2}(t),\ldots\big) = 
	\Psi_k \circ \cdots \circ \Psi_1 \circ
	\underbrace{\phi^{(N,\alpha)}\circ \cdots \circ \phi^{(N,\alpha)}}_{m\,\,\mathrm{times}}(a),
	\end{equation}
where $\Psi_j = \xi \circ \sigma_+^{\Delta i_j-1}+\phi^{(N_j,\alpha)} \circ \sigma_+^{\Delta i_j}$ with $N_j = N-i_j$ and $\Delta i_j = i_j-i_{j-1}$.
\end{theorem}

\begin{proof}
The last $m$ terms in Eq.~(\ref{eq3_GS}) determine the evolution from $a = u(0)$ to $u(m)$ over the integer time $m$. Thus, it remains to understand the $\Psi_j$ terms. 
Figure~\ref{fig4} shows an example for the special case $ m = 0$, $i_1 = 1$ and $i_2 = 2$ corresponding to time $t = 3/4$. 
Let us give a proof for this case, from which the general proof can be deduced by analogy.
In this case Eq.~(\ref{eq3_GS}) becomes
	\begin{equation}
	\label{eq3_GSex}
	\Big(u_2\big(\textstyle\frac{3}{4}\big),u_3\big(\textstyle\frac{3}{4}\big),u_4\big(\textstyle\frac{3}{4}\big),\ldots\Big) = 
	\Psi_2 \circ \Psi_1(a)
	\end{equation}
with 
	\begin{equation}
	\label{eq3_GSpsi}
	\Psi_1 = \xi+\phi^{(N-1,\alpha)} \circ \sigma_+, \quad
	\Psi_2 = \xi+\phi^{(N-2,\alpha)} \circ \sigma_+.
	\end{equation}
From Fig.~\ref{fig4} one can see that the operators (\ref{eq3_GSpsi}) define the evolutions over the time intervals $\tau_1 = 1/2$ and $\tau_2 = 1/4$, respectively. 
Note that the first term $\xi$ in the operators (\ref{eq3_GSpsi}) correspond to the energy transfer from the upper scales. The second terms describe the remaining transfers, which have the same form as in systems regularized at the scales $N-1$ and $N-2$.
Hence, the composition of $\Psi_1$ and $\Psi_2$ yields the state at $t = 3/4$.
\end{proof}

\begin{figure}[tp]
\centering
\includegraphics[width=0.7\textwidth]{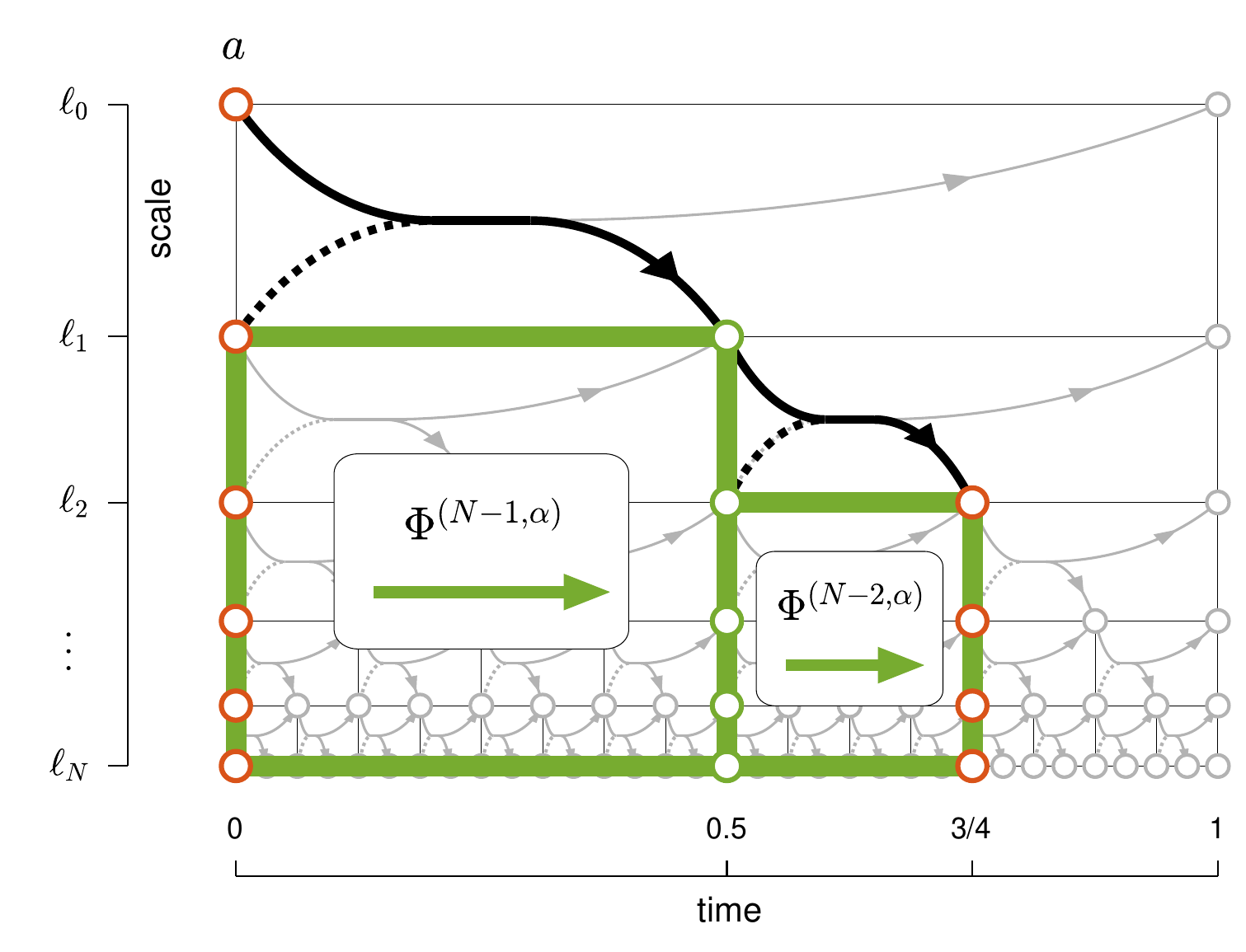}
\caption{Schematic representation of the relations (\ref{eq3_GSex}) and (\ref{eq3_GSpsi}).} 
\label{fig4}
\end{figure}

Theorem~\ref{th1} suggests a natural generalization for the inviscid limit of regularized dynamics. 
Let $\phi^{(0)}$ be an arbitrary flow map representing the system regularized at scale $N = 0$. 
We then introduce solutions regularized at scale $N$ using the flow maps
	\begin{equation}
	\label{eq2_RGG}
	\phi^{(N)} = \mathcal{R}^N[\phi^{(0)}]. 
	\end{equation}
Given an initial condition $a$, the regularized solution in unit time is defined as $u(1) = \phi^{(N)}(a)$.
Then the regularized solution $\{u_n(t)\}$ on the whole lattice is determined by the relations of Theorem~\ref{th2}, where $\alpha$ is dropped from the superscripts.
By construction, these solutions satisfy the initial conditions (\ref{eq2_5}) and the ideal equations (\ref{eq2_2}) for $n < N$, while the dynamics at scales $n \ge N$ is governed by the relations induced by the initial map $\phi^{(0)}$. 
We do not specify exactly which flow maps $\phi^{(0)}$ are admissible for (the domain of) the RG operator.
Such a rigorous construction is beyond the scope of this paper, where we only formulate assumptions supported by a collection of specific examples.
  
In summary, the RG operator (\ref{eq3_5}) relates the inviscid limit with the RG dynamics in the space of flow maps. 
In this scenario, the choice of the regularization translates into the initial flow map $\phi^{(0)}$, while the RG operator is defined uniquely in terms of the ideal transfer function $f$. This distinction of the regularization and ideal dynamics in the RG formalism is essential: it suggests that the inviscid limit is controlled by RG attractors, which are determined by the ideal system. Then, the choice of regularization selects one of these attractors by setting the initial flow map $\phi^{(0)}$ in a respective basin of attraction. 

Finally, we notice a similarity of our RG operator (\ref{eq3_5}) with the Feigenbaum equation~\cite{feigenbaum1983universal}. 
The latter involves the composition of the map with itself, combined with scaling. 
In our system, this double composition arises from the doubling of the turnover time between the $N$ and $N+1$ viscous scales.

\section{Fixed-point RG attractor}
\label{sec4}

In this section we formulate the concept of a fixed-point attractor in RG dynamics (\ref{eq2_RGG}) using an analogy with asymptotically stable equilibria in dynamical systems.
From a mathematical point of view, this analogy is not straightforward, since the RG operator acts in an infinite-dimensional space of flow maps. 
Our goal here is not to establish a rigorous theory, but to provide convincing numerical evidence that the RG dynamics follows the classical scenario of dynamical systems, and to draw conclusions from this for the inviscid limit. \newtext{We recall that the goal of our theory is understanding the finite-time predictability in the inviscid limit $N \to \infty$. This is different from the study of long-time statistics in turbulence, which refers to the limit $t \to \infty$.}

\subsection{Attractor and linearized RG dynamics}

The fixed-point RG attractor $\phi^\infty$ is the flow map obtained in the limit
	\begin{equation}
	\label{eq4_lim}
	\phi^\infty = \lim_{N \to \infty}\phi^{(N)} = \lim_{N \to \infty}\mathcal{R}^N[\phi^{(0)}], \quad
	\phi^{(0)} \in \mathcal{B}(\phi^\infty),
	\end{equation}
where $\mathcal{B}(\phi^\infty)$ denotes a basin of attraction in the space of flow maps.
As pointed out in~\cite{mailybaev2023spontaneous}, it is convenient to understand the limit in Eq.~(\ref{eq4_lim}) as the convergence $\phi^{(N)}(a^{(N)}) \to \phi^{\infty}(a)$ for any converging sequence $a^{(N)} \to a$. 
Here the convergence in the phase space is considered in the weak (product topology) sense, i.e., $a^{(N)} \to a$ means that the components $a_n^{(N)} \to a_n$ converge at every scale $n$.
According to this definition, the inviscid limit commutes with the composition of flow maps in the formulation of the RG operator (\ref{eq3_5}).
Therefore, taking the limit on both sides of the relation $\phi^{(N+1)} = \mathcal{R}[\phi^{(N)}]$, we find that the attractor $\phi^\infty$ is the fixed-point of the RG operator:
	\begin{equation}
	\label{eq4_FP}
	\phi^\infty =\mathcal{R}[\phi^\infty].
	\end{equation}
Moreover, the regularized solutions $\{u_n(t)\}$, which are defined in terms of the flow maps $\phi^{(N)}$ by the relations of Theorem~\ref{th2}, converge pointwise to solutions of the ideal initial value problem \cite{mailybaev2023spontaneous}. Thus, the convergence in Eq.~(\ref{eq4_lim}) predicts that the inviscid limit of regularized solutions exists and it is universal for all regularizations from the basin of attraction, $\phi^{(0)} \in \mathcal{B}(\phi^\infty)$. 

We further consider the linearization of the RG operator about the fixed-point $\phi^\infty$ as
	\begin{equation}
	\label{eq4_RGL1}
	\mathcal{R}(\phi^\infty+\delta\phi) 
	\approx \phi^\infty+\delta \mathcal{R}^\infty[\delta\phi].
	\end{equation}
Here $\delta \mathcal{R}^\infty[\delta\phi]$ denotes the variational derivative at the fixed point $\phi^\infty$, the explicit expression of which follows from the definition (\ref{eq3_5}) as
	\begin{equation}
	\label{eq4_RGL2}
	\delta \mathcal{R}^\infty[\delta\phi] 
	= \sigma_- \circ \delta\phi \circ (\xi+\phi^\infty \circ \sigma_+)+
	\sigma_- \circ \delta\phi^\infty \left(\xi+\phi^\infty \circ \sigma_+;\,\delta \phi \circ \sigma_+\right),
	\end{equation}
assuming that the map $\phi^\infty(a)$ has a variational derivative $\delta\phi^\infty(a;\delta a)$ at any state $a$. 

The convergence (\ref{eq4_lim}) of the RG dynamics (\ref{eq2_RGG}) together with the variation (\ref{eq4_RGL1}) yield a linearized formulation in the neighborhood of the attractor $\phi^\infty$ as
	\begin{equation}
	\label{eq4_LR}
	\delta\phi^{(N+1)} \approx \delta \mathcal{R}^\infty[\delta\phi^{(N)}], \quad 
	\delta\phi^{(N)} = \phi^{(N)}-\phi^\infty,
	\end{equation}
which is valid asymptotically for large $N$. As in the classical stability theory~\cite{arnold1992ordinary}, we assume that the linearized RG operator has a leading (largest absolute value) eigenvalue $\rho$. This eigenvalue together with the corresponding eigenvector (map) $\psi:\mathbb{U} \mapsto \mathbb{U}$ satisfy the eigenvalue problem 
	\begin{equation}
	\label{eq4_EV}
	\delta \mathcal{R}^\infty[\psi] = \rho \psi.
	\end{equation}
Then, the asymptotic form of the deviations (\ref{eq4_LR}) for large $N$ takes the form
	\begin{equation}
	\label{eq4_D}
	\delta\phi^{(N)} \approx c \rho^N \psi,
	\end{equation}
where $c$ is a real coefficient depending on the initial flow map $\phi^{(0)}$. Expression (\ref{eq4_D}) is formulated for a real eigenmode.
If the eigenmode (including $\rho$, $\psi$ and $c$) is complex, then the real part should be taken on the right-hand side of Eq.~(\ref{eq4_D}). 

The asymptotic relation (\ref{eq4_D}) implies a more subtle universality of the inviscid limit. Indeed, both the eigenvalue $\rho$ and the eigenvector $\psi$ are defined by the linearized RG operator and, therefore, do not depend on the choice of regularization. The choice of regularization, i.e. the choice of the initial map $\phi^{(0)}$, affects only the coefficient $c$ in the asymptotic relation (\ref{eq4_D}).

\subsection{Numerical example}

We now verify the described fixed-point scenario numerically in a specific example. Let us take the regularized transfer functions (\ref{eq2_7}) with
	\begin{equation}
	\label{eq4_1}
	f(u,v) = \left\{ \begin{array}{ll} 
	0.3-0.1 \cos \left(5 e^{-u/v} \right), & v > 0; \\
	0.2, & v = 0.
	\end{array} \right.
	\end{equation}
We solve the corresponding regularized system numerically and compute the flow maps $\phi^{(N,\alpha)}(a)$ for the specific initial condition
	\begin{equation}
	\label{eq4_1IC}
	a_n = \left\{\begin{array}{ll}
	1-n/5, & n = 0,\ldots,4; \\ 0, & n \ge 5.
	\end{array}\right.
	\end{equation}
Figure~\ref{fig5}(a) shows that the regularized states $\phi^{(N,\alpha)}(a)$ converge to the same inviscid limit $\phi^\infty(a)$ for large $N = 12,\ldots,15$ and two different parameters $\alpha = 1/4$ and $3/4$; these states are visually undistinguishable in the figure.

\begin{figure}[tp]
\centering
\includegraphics[width=0.9\textwidth]{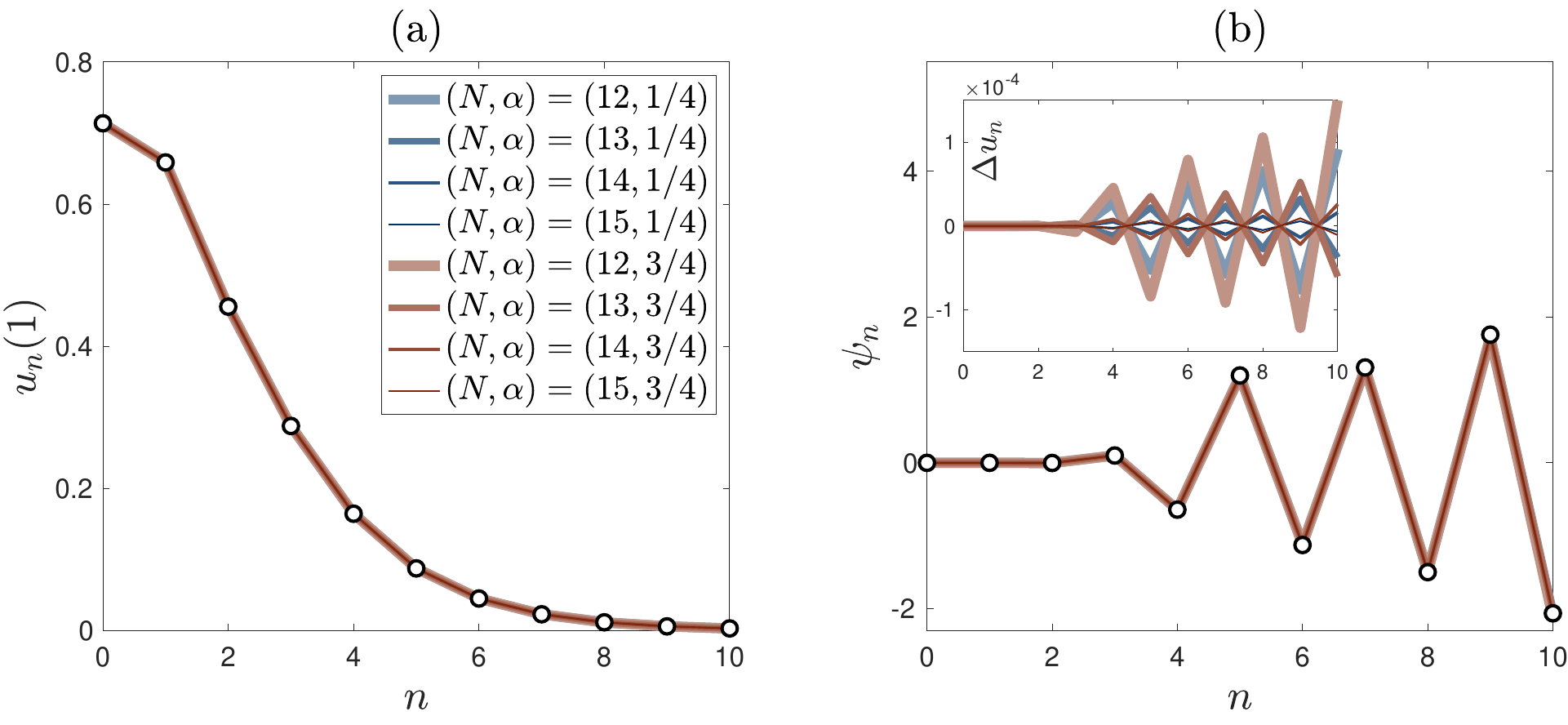}
\caption{(a) Convergence (very accurate collapse) of the states $u(1) = \phi^{(N,\alpha)}(a)$ for large $N = 12,\ldots,15$ and regularization parameters $\alpha = 1/4$ and $3/4$. 
(b) Components of the corresponding eigenvector $\psi(a)$ computed by Eq.~(\ref{eq4_OmN}) for the different parameters $(N,\alpha)$ (same as in the left panel); all the curves match very accurately. The inset shows the original differences $\Delta u = \phi^{(N+1,\alpha)}(a)-\phi^{(N,\alpha)}(a)$ before the rescaling.} 
\label{fig5}
\end{figure}

For the analysis of the leading eigenmode it is convenient to use the Cauchy differences 
	\begin{equation}
	\label{eq4_CD}
	\Delta u = \phi^{(N+1,\alpha)}(a)-\phi^{(N,\alpha)}(a) \approx c_\alpha (\rho-1)\rho^N \psi(a),
	\end{equation}
where the last asymptotic relation follows from Eq.~(\ref{eq4_D}). We wrote the coefficient as $c_\alpha$ because it depends on the initial flow map $\phi^{(0,\alpha)}$, i.e., on the parameter $\alpha$.  The Cauchy differences (\ref{eq4_CD}) are computed numerically and shown in the inset of Fig.~\ref{fig5}(b) for the same parameters $(N,\alpha)$ as in the panel (a). Using asymptotic relation (\ref{eq4_CD}), we estimate the eigenvalue as
	\begin{equation}
	\label{eq4_OmE}
	\rho \approx \frac{\left[\phi^{(N+1,\alpha)}(a)-\phi^{(N,\alpha)}(a)\right]_n}{\left[\phi^{(N,\alpha)}(a)-\phi^{(N-1,\alpha)}(a)\right]_n},
	\end{equation}
where $[\cdot]_n$ denotes the $n$-component of the sequence in parentheses, and we choose $n = 4$. 
Evaluating the expression (\ref{eq4_OmE}) with increasing $N$, we find the eigenvalue $\rho \approx -0.42$. The absolute value $|\rho| < 1$ implies asymptotic stability of the fixed point attractor, which is in agreement with our numerical results.

We express the eigenvector from Eq.~(\ref{eq4_CD}) as
	\begin{equation}
	\label{eq4_OmN}
	\psi(a) \approx
	\frac{\phi^{(N+1,\alpha)}(a)-\phi^{(N,\alpha)}(a)}{c_\alpha(\rho-1)\rho^N}.
	\end{equation}
We can set $c_{1/4} = 1$, which is equivalent to normalizing the eigenvector.
Then, the coefficient $c_{3/4} \approx 1.61$ is estimated numerically. The sequences computed by Eq.~(\ref{eq4_OmN}) are shown in Fig.~\ref{fig5}(b) for the same parameters $(N,\alpha)$ as in the panel (a). The accurate collapse of these graphs strongly supports our conjecture that the fixed-point RG attractor can be described in terms of classical stability theory.

We tested various initial conditions and found that all the conclusions made above are verified if the initial conditions decay sufficiently fast at small scales, namely, faster than (a multiple of) the turnover time $a_n \lesssim C\tau_n$. 
\newtext{From numerical measurements we observed that $u_n(1) \propto \tau_n$ for the limiting solution,} suggesting that the scaling $\propto \tau_n$ is a common property of solutions in the inviscid limit.
When the initial conditions are ``rougher'', e.g., $a_n \propto \tau_n^h$ as $n \to \infty$ with $h < 1$, then the universality of the inviscid limit is broken: for example, the eigenmode relation (\ref{eq4_D}) is no longer satisfied when $a_n \propto \tau_n^{1/2}$. 
We emphasize that these are just numerical observations, which may be useful for developing a more rigorous RG theory in the future.
They indicate that additional regularity conditions are necessary in defining the phase space $\mathbb{U}$.

\section{Period doubling bifurcation in RG dynamics}
\label{sec5}

We now extend our analysis to a family of ideal models depending on a real parameter $p$. These models are defined by the coupling function 
	\begin{equation}
	\label{eq5_1}
	f(u,v) = \left\{ \begin{array}{ll} 
	0.3-0.1 \cos \left(p e^{-u/v} \right), & v > 0; \\
	0.2, & v = 0;
	\end{array} \right.
	\end{equation}
which reduces to the previous case (\ref{eq4_1}) for $p = 5$. As the initial state we use Eq.~(\ref{eq4_1IC}).
In this section we focus on the loss of stability of the fixed-point RG attractor with subsequent bifurcation and its consequences for the inviscid limit.

Using the method from the previous section, we computed the eigenvalue $\rho$ numerically for different values of the parameter $p \ge 5$. The results are shown in Fig.~\ref{fig6}(a) by empty squares. We see that the eigenvalue decreases and attains the critical value $\rho = -1$ at $p_{\mathrm{pd}} \approx 6.95$. 

The fixed-point RG attractor becomes unstable when the parameter passes the point with eigenvalue $\rho = -1$.
In classical bifurcation theory~\cite{guckenheimer2013nonlinear}, this type of instability leads to the period-doubling bifurcation. In Fig.~\ref{fig6}(b) we confirm numerically that this bifurcation indeed takes place for the RG dynamics. In this figure, we plot the component $u_4(1)$ of the state $u(1) = \phi^{(N,\alpha)}(a)$ as a function of the parameter $p$. The regularization parameters are taken as $(N,\alpha) = (25,1/4)$ (shown by crosses) and $(N,\alpha) = (26,1/4)$ (shown by circles). 
Let us denote by $\Delta u_4(1)$ the difference between the values of $u_4(1)$ for $N = 25$ and $26$.
The inset in Fig.~\ref{fig6}(b) shows that the squared difference $[\Delta u_4(1)]^{2}$ grows linearly after the bifurcation, as predicted by classical bifurcation theory~\cite{guckenheimer2013nonlinear}. 
The thin curve and the red dot in Fig.~\ref{fig6}(b) are interpolations of the period-doubling bifurcation expected in the limit $N \to \infty$.

\begin{figure}[tp]
\centering
\includegraphics[width=0.95\textwidth]{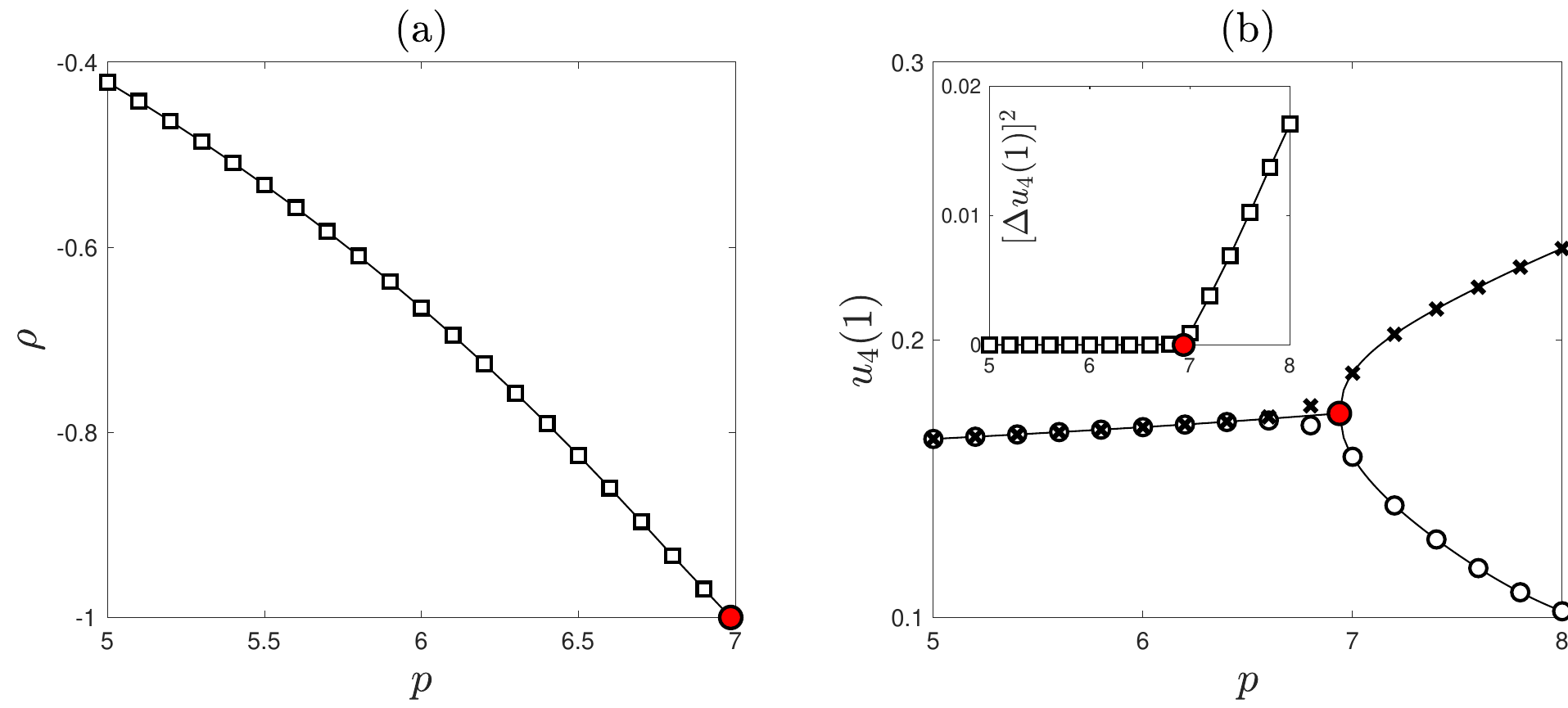}
\caption{(a) Dependence of the eigenvalue $\rho$ of the fixed-point RG attractor on the parameter $p$. In numerical calculations we used $(N,\alpha) = (25,1/4)$.
(b) Components $u_4(1)$ of the states $u(1) = \phi^{(N,\alpha)}(a)$ computed for different $p$ with the fixed $(N,\alpha) = (25,1/4)$ (shown by crosses) and $(N,\alpha) = (26,1/4)$ (shown by circles). The inset shows the squared difference $[\Delta u_4(1)]^{2}$. The red points and the thin lines show the bifurcation diagram expected in the limit $N \to \infty$.} 
\label{fig6}
\end{figure}
\begin{figure}[tp]
\centering
\includegraphics[width=0.55\textwidth]{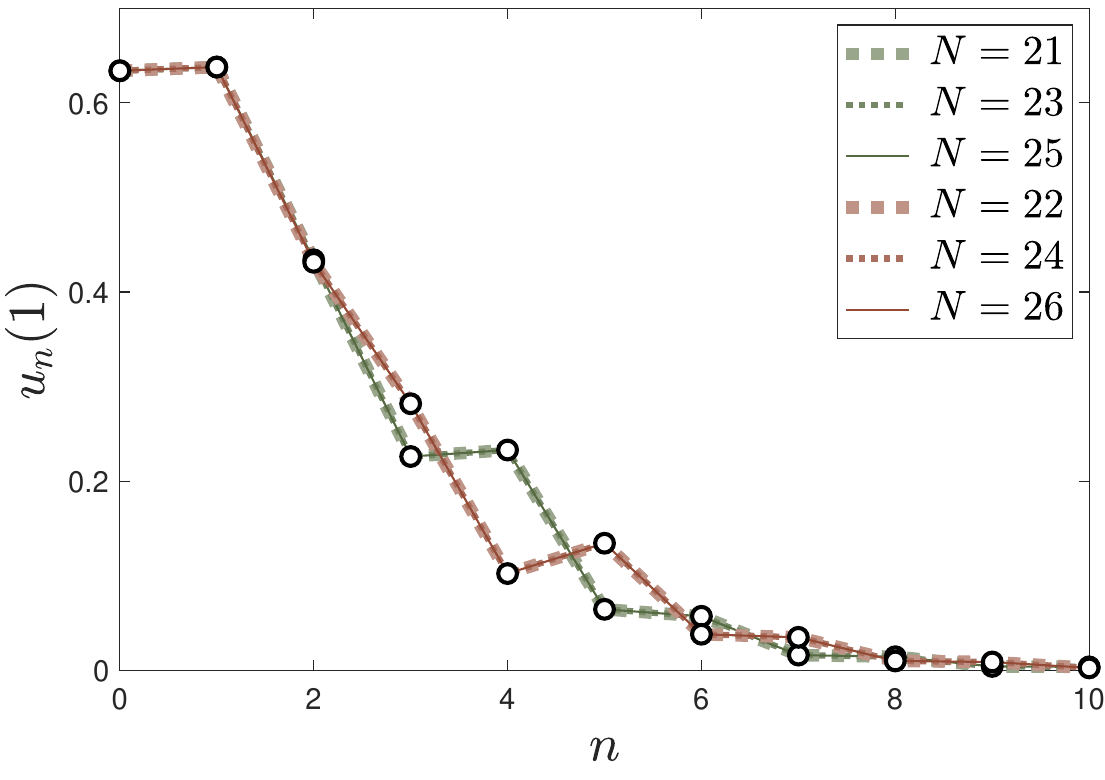}
\caption{The states $u(1) = \phi^{(N,\alpha)}(a)$ computed for $p = 8$, $\alpha = 1/4$ and different $N$. These graphs demonstrate convergence to distinct inviscid limits for odd and even $N$. The same results (not shown) are obtained for $\alpha = 3/4$.}
\label{fig7}
\end{figure}

Existence of the period-2 RG attractor at $p > p_{\mathrm{pd}}$ implies that the regularized system has two different behaviors in the inviscid limit depending on whether $N$ is even or odd; see Fig.~\ref{fig7}. 
Namely, there are two different limits 
	\begin{equation}
	\label{eq5_2}
	\phi_a^\infty = \lim_{N \to \infty}\phi^{(2N)}, \quad
	\phi_b^\infty = \lim_{N \to \infty}\phi^{(2N+1)}, \quad
	\phi^{(0)} \in \mathcal{B}(\phi_a^\infty,\phi_b^\infty),
	\end{equation}
where $\mathcal{B}(\phi_a^\infty,\phi_b^\infty)$ denotes the basin of attraction in the space of flow maps. The pair $\phi_a^\infty$ and $\phi_b^\infty$ is universal (up to a permutation) in the sense that it does not depend on regularization, provided that $\phi^{(0)}$ belongs to the basin of attraction. In other words, for large $N$ we observe one of the two solutions to the ideal initial value problem, which are defined by the flow maps $\phi_a^\infty$ and $\phi_b^\infty$ and depend only on the parity of $N$. 

\section{Regularization by noise and stochastic RG operator}
\label{sec6}

\subsection{Chaotic RG dynamics}

Let us study a different situation when the RG dynamics has no regular (fixed-point or periodic) attractor. We consider the same family of regularized systems given by Eqs.~(\ref{eq2_2}), (\ref{eq2_7}) and (\ref{eq2_6}), which depend on the regularization parameters $(N,\alpha)$. But now we use a different transfer function of the form
	\begin{equation}
	\label{eq6_1}
	f(u,v) = \left\{ \begin{array}{ll} 
	0.4-0.3 \cos \left(10.3 e^{-u/v}-5.15 \right), & v > 0; \\
	0.4-0.3 \cos 5.15, & v = 0.
	\end{array} \right.
	\end{equation}
For numerical analysis, we use the same initial state $a$ from Eq.~(\ref{eq4_1IC}).
Figure~\ref{fig8}(a) shows the states $u(1) = \phi^{(N,\alpha)}(a)$ for $\alpha = 1/4$ and different regularization scales $N$. 
This time, no convergence is observed as $N \to \infty$, suggesting that the RG dynamics may be chaotic.
In Fig.~\ref{fig8}(b) we performed a conventional chaos test: we measure how the flow maps diverge in the RG dynamics (i.e. with increasing $N$) by introducing a small  perturbation of the initial flow map at $N = 0$.
We define this perturbation as a tiny change $\delta \alpha = 10^{-15}$ of the regularization parameter. 
Here we compute the differences $\delta u = \phi^{(N,\alpha+\delta\alpha)}(a)-\phi^{(N,\alpha)}(a)$ and the graph in Fig.~\ref{fig8}(b) shows the norms $\|\delta u\| = \sqrt{\delta u_0^2+\delta u_1^2+\cdots}$. We indeed see a very rapid increase of deviations. 
However, in contrast to classical chaos, which is characterized by exponential growth of deviations, in our case the increase of $\|\delta u\|$ is superexponential. The inset in Fig.~\ref{fig8}(b) shows the graph of $\log \log (\|\delta u\|/\delta\alpha)$, suggesting double exponential behavior. \newtext{A possible explanation to this double exponential growth is that the separation $\|\delta u\| \propto \delta\alpha \,e^{\lambda_{\max}}$ at time $t = 1$ is exponential with respect to the largest Lyapunov exponent of the flow map. Assuming that $\lambda_{\max}$ is determined by the smallest active scale, the dimensional prediction yields the second exponential relation as $\lambda_{\max} \sim 1/\tau_N = 2^N$.}

\begin{figure}[tp]
\centering
\includegraphics[width=0.9\textwidth]{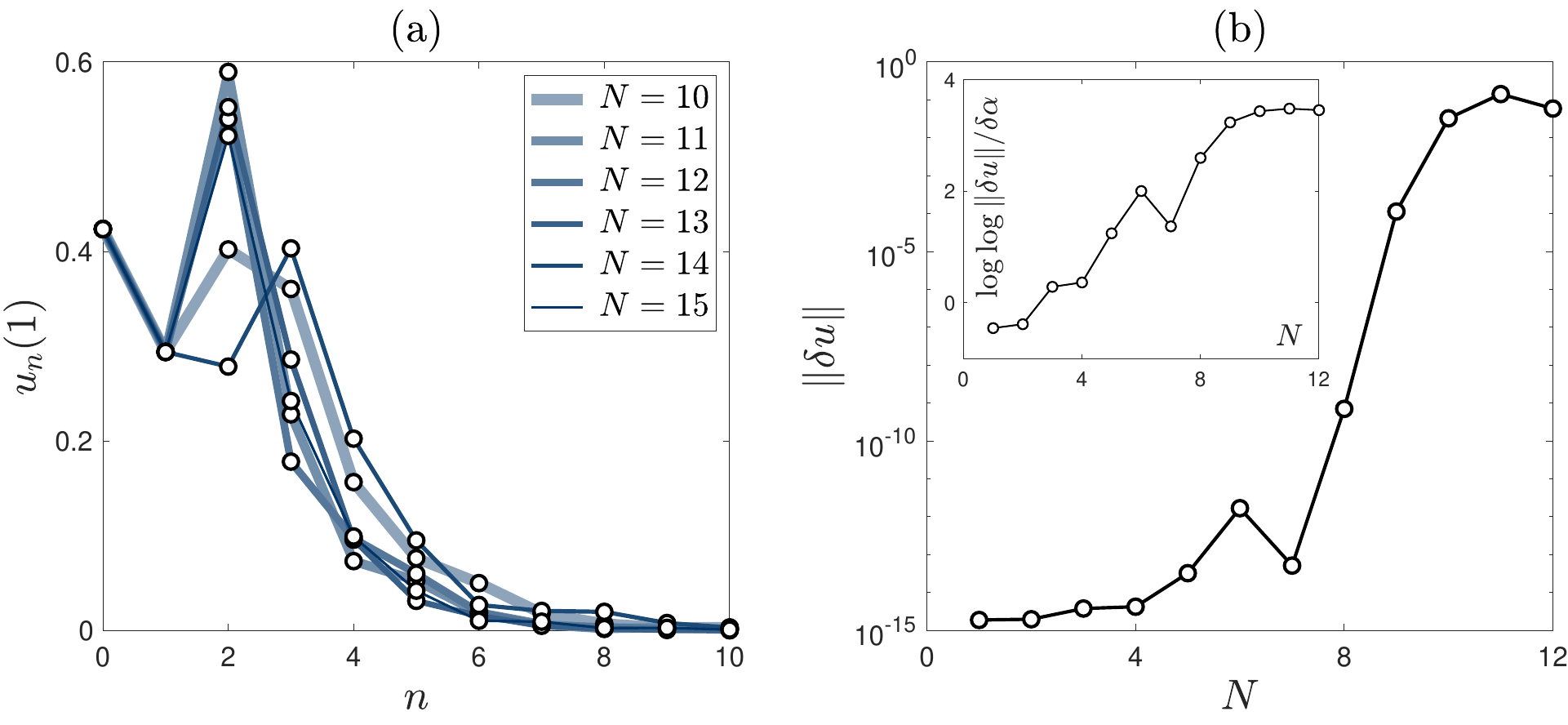}
\caption{(a) Irregular RG dynamics for the transfer function (\ref{eq6_1}). (b) Superexponential increase of deviations in the RG dynamics shown in logarithmic vertical scale. The inset plots the double logarithm $\log \log (\|\delta u\|/\delta\alpha)$.}
\label{fig8}
\end{figure}

In classical chaos (e.g. for hyperbolic systems~\cite{eckmann1985ergodic}) the distribution of a random ensemble of initial conditions evolves at large times to the probability measure of a chaotic attractor. 
Here, even a small amount of initial randomness is sufficient due to the exponential amplification of small perturbations, widely known as the butterfly effect.
Since this amplification is even stronger in our RG dynamics, we follow the same strategy: we introduce uncertainties into the RG dynamics. \newtext{We do this by adding noise to the system at the viscous scale $N$ so that the randomness disappears in the inviscid limit $N \to \infty$. Another point of view underlying the probabilistic approach~\cite{flandoli2008markov,vishik2012mathematical} is the lack of uniqueness in the ideal initial value problem.}

\subsection{Stochastic RG operator}

In hydrodynamic turbulence, a natural source of uncertainties is microscopic noise, which strongly influences the viscous scale of motion due to very large Lyapunov exponents~\cite{ruelle1979microscopic}. 
We will mimic this situation in our model by introducing noise at the viscous scale $N$.
Specifically, we define the stochastic regularization by randomizing the transfer function (\ref{eq2_7}) as
	\begin{equation}
	\label{eq6_2}
	f_n(t) = \left\{ \begin{array}{ll}
	f\big(u_n(t),u_{n+1}(t)\big), & n < N; \\
	\alpha_t, & n = N.
	\end{array}\right.
	\end{equation}
Here the values  of $\alpha_t$ at different times $t/\tau_N = 0,1,2,\ldots$ are independent and identically distributed (i.i.d.) random numbers with a given probability distribution. 
As a result, the stochastic term $(1-\alpha_t) u_N$ appears in the dynamics Eq.~(\ref{eq2_2}) at the viscous scale $N$, representing small-scale noise.

Given the initial condition $u(0) = a$, the solution $u(1)$ at the turnover time $\tau_0 = 1$ becomes a random variable. 
Then, instead of the deterministic flow map (\ref{eq3_1}), it is convenient to introduce a probability (Markov) flow kernel
	\begin{equation}
	\label{eq6_3}
	\Phi^{(N)}(A|a).
	\end{equation}
For fixed $a$, this is a probability  measure  determining the probability that the solution $u(1) \in A$ for a measurable subset $A \subset \mathbb{U}$. 
In other words, the flow kernel (\ref{eq6_3}) determines the probability of transition from the initial state $u(0) = a$ to different states $u(1)$ per unit time.
Following~\cite{mailybaev2023spontaneous}, where a similar analysis is performed for a slightly different system, we now rewrite the RG relation between flow maps from Section~\ref{sec3} as the RG relation between flow kernels.

First, let us recall operations with probability kernels (see, e.g.,~\cite{taylor2001user}), which replace respective operations (composition and addition) of deterministic maps. 
For two kernels $\Phi$ and $\Phi'$, their composition is a kernel $\Phi' \circ \Phi$ defined as 
	\begin{equation}
	\label{eqSK3}
	\Phi' \circ \Phi(A|a) 
	=  \int \Phi'(A|b)\Phi(db|a),
 	\end{equation}
where the integration is with respect to $db$. 
Relation (\ref{eqSK3}) defines the probability distribution for the composition $\phi' \circ \phi(a)$, where $\phi$ and $\phi'$ are statistically independent random maps whose distributions are given by $\Phi$ and $\Phi'$. 
Similarly, the probability distribution of the sum $\phi(a)+\phi'(a)$ is described by the convolution kernel $\Phi \ast \Phi'$ defined as
	\begin{equation}
	\label{eqSK5}
	\Phi \ast  \Phi' (A|a) 
	= \iint \chi_{A}(u'+u'') \Phi(du'|a) \Phi(du''|a),
 	\end{equation}
where the integration is with respect to $du'$ and $du''$, and $\chi_{A}(\cdot)$ is the characteristic function. 

From the derivations of Section~\ref{sec3} (see also~\cite{mailybaev2023spontaneous}) one can see that the deterministic RG relations (\ref{eq3_4}) and (\ref{eq3_5}) extend to the case of stochastic regularizations by replacing all maps with the corresponding probability kernels, and all operations between maps (composition and addition)  with their analogues for kernels  (composition and convolution). 
Then, the RG relation between the flow kernels (\ref{eq6_3}) becomes
	\begin{equation}
	\label{eq6_R1}
	\Phi^{(N+1)} = \mathfrak{R}[\Phi^{(N)}],
	\end{equation}
where we introduce the stochastic RG operator $\mathfrak{R}$ as
	\begin{equation}
	\label{eq6_R2}
	\mathfrak{R} [\Phi] = \Delta_{\pi_0-\xi} \ast \big(\Delta_{\sigma_-} \circ \Phi \circ (\Delta_\xi \ast \Phi \circ \Delta_{\sigma_+})\big).
	\end{equation}
Here $\Delta_\xi$ denoted the Dirac kernel of the deterministic map $\xi$, defined in terms of the Dirac measures as $\Delta_\xi(A|a) = \delta_{\xi(a)}(A)$. 

The expressions of Theorem~\ref{th2} are similarly extended for stochastic regularizations~\cite{mailybaev2023spontaneous}. Hence, probability distributions of stochastically regularized solutions at any time can be expressed in terms of the flow kernels $\Phi^{(N)}$.

As in the deterministic case, the stochastic RG operator (\ref{eq6_R1}) relates the change of the viscous scale $N$ in the stochastically regularized system with the RG dynamics in the space of flow kernels.
In this scenario, the choice of both regularization and noise translates into the definition of an initial flow kernel $\Phi^{(0)}$, while the transfer function $f$ of the ideal system determines the RG operator (\ref{eq6_R2}) through the map $\xi$ from  Eq.~(\ref{eq3_3}). 
The inviscid limit $N \to \infty$ is thus associated with RG attractors that depend only on the ideal system. 
The choice of regularization selects one of these attractors by setting the initial flow kernel $\Phi^{(0)}$ in the corresponding basin of attraction.

\section{Spontaneous stochasticity as RG attractor}
\label{sec7}

We now apply the RG formalism to the stochastically regularized system with the transfer functions given by Eqs.~(\ref{eq6_1}) and (\ref{eq6_2}). 
For the i.i.d. random variables (noise) $\alpha_t$ we use a uniform probability distribution in the interval $\alpha \in [0.4,0.5]$, denoted as $\mu$. 
As the second type of stochastic regularization, we use a uniform distribution in the much smaller interval $\alpha \in [0.3,0.301]$, denoted as $\tilde\mu$. 
In the second case, random fluctuations are significantly smaller than the deterministic (``viscous'') part, simulating a similar property of molecular noise in turbulence.
The initial condition $a$ is the same as before and is given by Eq.~(\ref{eq4_1IC}).
We analyze numerically the flow  kernels $\Phi^{(N)}$ by  computing the transition probabilities $a \mapsto u(1)$ using Monte Carlo sampling.  
For this purpose, we compute $10^6$ random states $u(1)$ for each viscous scale $N = 16,17,18$. 
Figure~\ref{fig9} shows the probability densities of components $u_2(1)$, $u_3(1)$ and $u_4(1)$ computed using the histogram method. 
The graphs collapse with high accuracy, confirming that the PDFs converge to nontrivial probability distributions as  $N \to \infty$,  and  these distributions  do not depend on the noise  ($\mu$ or $\tilde\mu$). 
We remark that, \newtext{according to numerical measurements,} the variables $u_0(1)$ and $u_1(1)$ turn out to be deterministic.

\begin{figure}[tp]
\centering
\includegraphics[width=0.99\textwidth]{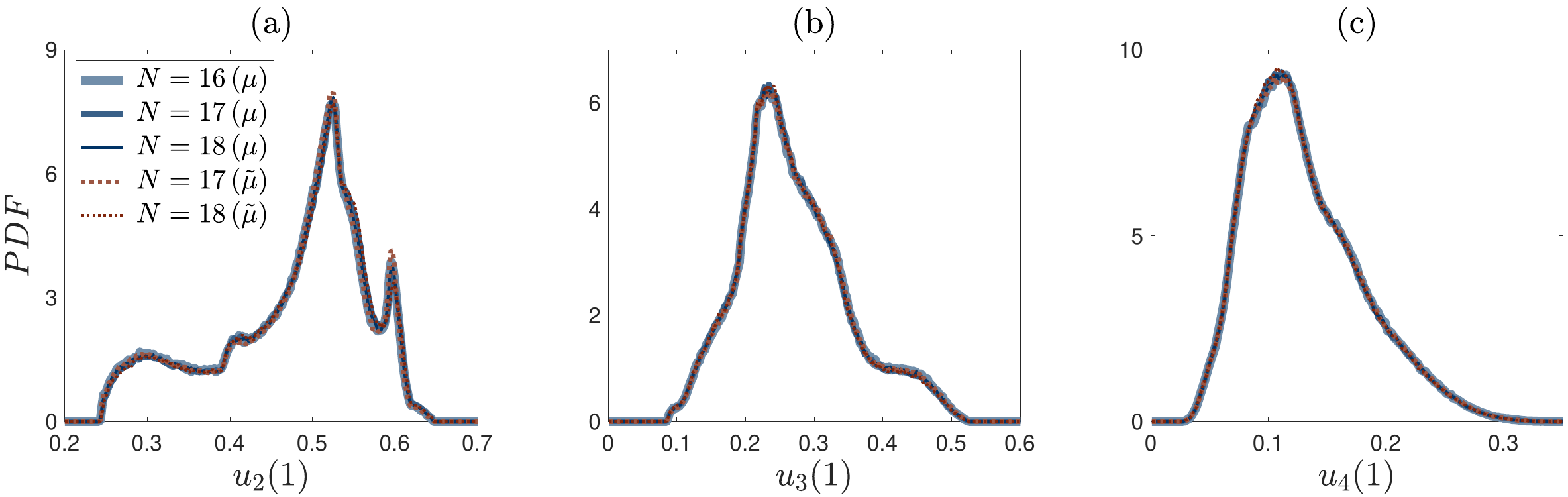}
\caption{PDFs of the state components (a) $u_2(1)$, (b) $u_3(1)$ and (c) $u_4(1)$ computed with $10^6$ random simulations. The graphs correspond to different $N = 16,17,18$ and two types of noise, $\mu$ and $\tilde\mu$. 
The accurate collapse of these graphs confirms the convergence to a fixed-point (spontaneously stochastic) attractor of the RG operator.}
\label{fig9}
\end{figure}

\subsection{Fixed-point RG attractor} 

The numerical results suggest that the RG dynamics has a fixed-point attractor in the space of flow kernels. In this case we replace the previous deterministic limit (\ref{eq4_lim}) by
	\begin{equation}
	\label{eq7_lim}
	\Phi^\infty = \lim_{N \to \infty}\Phi^{(N)} = \lim_{N \to \infty}\mathfrak{R}^N[\Phi^{(0)}], \quad
	\Phi^{(0)} \in \mathfrak{B}(\Phi^\infty),
	\end{equation}
where $\mathfrak{B}(\Phi^\infty)$ denotes a basin of attraction in the space of flow kernels. 
A natural definition~\cite{karr1975weak} in which the limit (\ref{eq7_lim}) can be understood requires that
	\begin{equation}
	\label{eqSlim}
	\int \varphi({{u}})\Phi^\infty(du|a) = 
	\lim_{N \to \infty} \int \varphi({{u}})\Phi^{(N)}(du|a^{(N)})
 	\end{equation}
for any converging sequence ${a}^{(N)} \to {a}$ of initial conditions and any bounded and uniformly continuous function $\varphi(u)$ depending on a finite number of state components $u_n$.
This definition is consistent with our numerical tests. 
Additionally, we checked numerically that the PDFs in Fig.~\ref{fig9} remain the same (up to small statistical fluctuations) if the fixed initial state is replaced by a converging sequence ${a}^{(N)} \to {a}$. 

The definition (\ref{eqSlim}) ensures that the limiting flow kernel (\ref{eq7_lim}) is a fixed point of the stochastic RG operator~\cite{mailybaev2023spontaneous}
	\begin{equation}
	\label{eqSS2b}
	\Phi^\infty = \mathfrak{R}[\Phi^\infty].
 	\end{equation}
The limiting kernel $\Phi^\infty$ corresponds to a system in which the regularization (including the noise) is removed (transferred to infinitesimal scales) as $N \to \infty$. 
In fact, using the relations of Theorem~\ref{th2} extended to the stochastic case, one can show~\cite{mailybaev2023spontaneous} that stochastically regularized solutions $\{u_n^{(N)}(t)\}$ determined by $\Phi^{(N)}$ converge weakly to a stochastic process $\{u_n^{\infty}(t)\}$ determined by $\Phi^\infty$. Furthermore, this limiting stochastic process is supported on solutions of the ideal initial value problem. 
Recall that the ideal initial value problem is formally deterministic, but may have nonunique solutions. Thus, the limit (\ref{eq7_lim}) with a nontrivial (not Dirac) flow kernel $\Phi^\infty$ imposes a random selection of these nonunique solutions.
The phenomenon just described (stochasticity that persists in the limit of vanishing noise) is known in turbulence as Eulerian spontaneous stochasticity~\cite{mailybaev2016spontaneously,thalabard2020butterfly,bandak2024spontaneous} for velocity fields and the Lagrangian spontaneous stochasticity for particle trajectories~\cite{falkovich2001particles}. 

In our RG formalism, spontaneous stochasticity is related to the fixed-point attractor $\Phi^\infty$ of the stochastic RG operator $\mathfrak{R}$. This explains the universality of spontaneously stochastic solutions: they are independent of regularization and noise, provided that the initial flow kernel $\Phi^{(0)}$ belongs to the basin of attraction.

\subsection{Local RG dynamics}

As a next step, we conjecture that the stochastic RG attractor $\Phi^\infty$ has a regular local structure. 
This means that the RG operator can be linearized about the fixed point as 
	\begin{equation}
	\label{eq7_RGL1}
	\mathfrak{R}(\Phi^\infty+\delta\Phi) 
	\approx \Phi^\infty+\delta \mathfrak{R}^\infty[\delta\Phi],
	\end{equation}
where $\delta \mathfrak{R}^\infty[\delta\Phi]$ is a variational derivative.
We further assume that the linearized RG operator $\delta \mathfrak{R}^\infty$ has a dominant (largest absolute value) eigenvalue $\rho$ and a corresponding eigenvector $\Psi$ that satisfy the eigenvalue problem
	\begin{equation}
	\label{eq7_EV}
	\delta \mathfrak{R}^\infty[\Psi] = \rho \Psi.
	\end{equation}
Here the eigenvector is a kernel $\Psi(A|a)$, which is a signed measure with  zero total mass $\Psi(\mathbb{U}|a) = 0$ for each initial state $a$.
Assuming that the eigenmode (\ref{eq7_EV}) is real, deviations of flow kernels $\Phi^{(N)} = \mathfrak{R}^N[\Phi^{(0)}]$ from the  limit  $\Phi^\infty$ take the asymptotic form
	\begin{equation}
	\label{eq7_D}
	\delta\Phi^{(N)} = \Phi^{(N)}-\Phi^\infty \approx  c \rho^N \Psi, \quad N \to \infty.
	\end{equation}
Since the eigenmode is a property of the RG operator, both $\rho$ and $\Psi$ are independent of the choice of regularization and noise.
Therefore, expression (\ref{eq7_D}) establishes the asymptotic universality of the deviations up to the constant factor $c$.  This factor is the only asymptotic property depending on the regularization.

Numerical testing of the relation (\ref{eq7_D}) is more difficult, since small differences between converging probability distributions must be measured.
For this purpose, we used the Monte Carlo method with $10^7$ simulations for $N = 14,\ldots,18$, the same initial state and the first type of noise. For numerical analysis it is convenient to use Cauchy differences
	\begin{equation}
	\label{eq7_CD}
	\Delta \Phi^{(N)} = \Phi^{(N+1)}-\Phi^{(N)} \approx \rho^N \Psi,
	\end{equation}
where the last asymptotic relation follows from Eq.~(\ref{eq7_D}); for simplicity, the constant factor $c(\rho-1)$ is absorbed by the normalization of the eigenvector $\Psi$.
Equation (\ref{eq7_CD}) yields the asymptotic expression for the eigenvector as
	\begin{equation}
	\label{eq7_CDb}
	\Psi \approx \frac{\Delta \Phi^{(N)}}{\rho^N}.
	\end{equation}
	
We access the signed measure $\Delta \Phi^{(N)}(A|a)$ for the given initial state $a$ numerically as follows. 
We compute the probability densities $p^{(N)}(x)$ and $p^{(N+1)}(x)$ of the same component $x = u_n(1)$ for the two regularization scales $N$ and $N+1$; see Fig.~\ref{fig9}. 
Their difference $\Delta p^{(N)}(x) = p^{(N+1)}(x)-p^{(N)}(x)$ represents the respective marginal density of $\Delta \Phi^{(N)}$. Using Eq.~(\ref{eq7_CD}) 
we estimate the eigenvalue $\rho \approx -0.7$ by minimizing the difference between rescaled functions $\Delta p^{(N)}(x)/\rho^N$ for different $N$. 
The resulting rescaled functions approximate the eigenvector  in Eq.~(\ref{eq7_CDb}) and are shown in Fig.~\ref{fig10}.  
A rather accurate (up to statistical fluctuations) collapse of these graphs for different $N$ verifies our conjecture that the spontaneously stochastic RG attractor has a local structure described in terms of the classical stability  theory.

\begin{figure}[tp]
\centering
\includegraphics[width=0.85\textwidth]{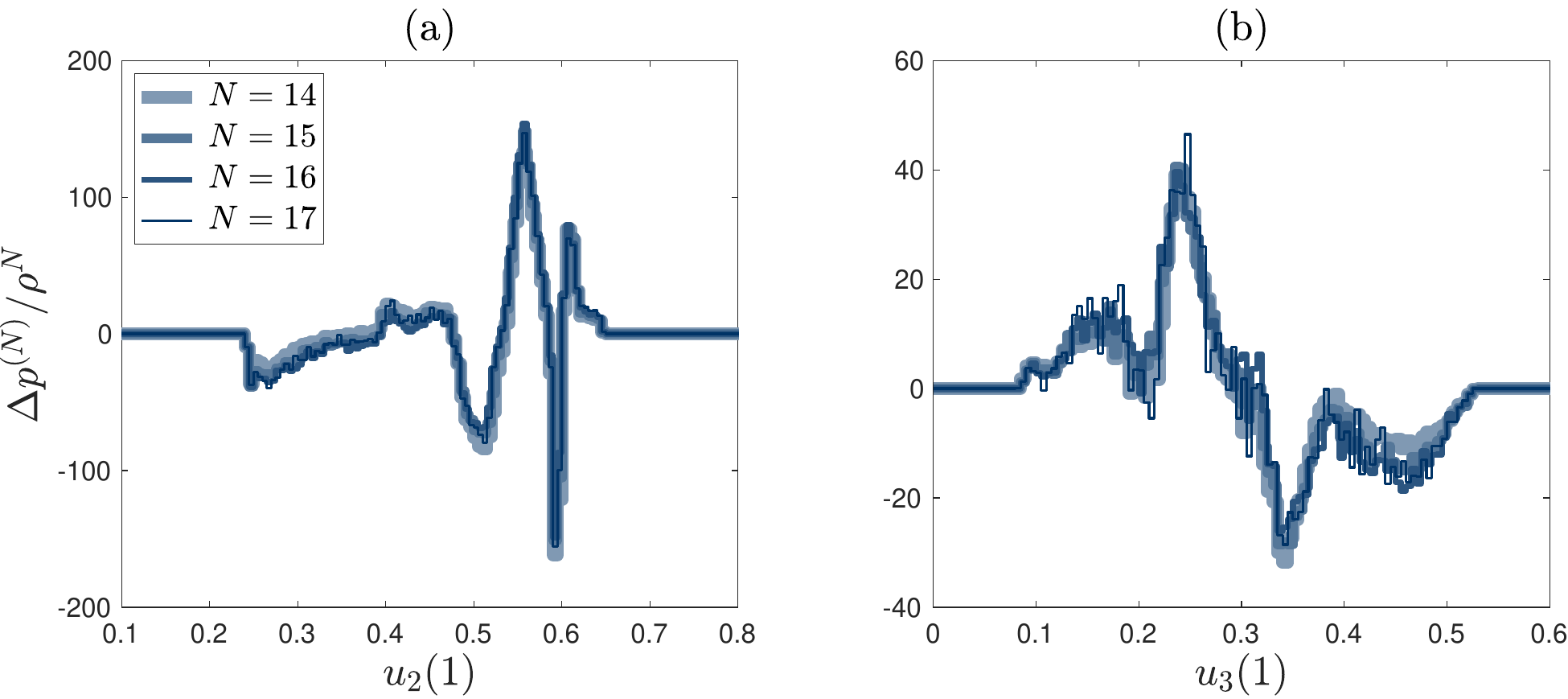}
\caption{Rescaled differences between PDFs of the state component (a) $u_2(1)$ and (b) $u_3(1)$, computed using $10^7$ random samples. The graphs correspond to $N = 14,\ldots,17$ and the first type of noise $\mu$. Their collapse verifies the asymptotic eigenmode relation (\ref{eq7_D}).}
\label{fig10}
\end{figure}

\section{Bifurcations of spontaneously stochasticity}
\label{sec8}

\newtext{The transition from deterministic to stochastic RG dynamics can be viewed as an analogue of the relation between deterministic and random dynamic systems. However, there are significant differences. In random dynamic systems, the evolution of the probability distribution is linear: recall, for example, the Fokker--Planck equation for the probability density. In contrast, the stochastic RG operator considered in our theory is nonlinear with respect to the flow kernel, as can be seen from Eq.~(\ref{eq6_R2}). Therefore, the stochastic RG operator induces nonlinear dynamics in the space of flow kernels. In this section, we demonstrate one consequence of this nonlinearity, namely, the period-doubling bifurcation of spontaneous stochasticity.}

Let us study how the fixed-point attractor of the stochastic RG operator changes and eventually bifurcates with a change of the transfer function. We consider
	\begin{equation}
	\label{eq8_1}
	f(u,v) = \left\{ \begin{array}{ll} 
	0.4-0.3 \cos \left(p e^{-u/v}-p/2 \right), & v > 0; \\
	0.4-0.3 \cos \left(p/2\right), & v = 0;
	\end{array} \right.
	\end{equation}
which yields the case (\ref{eq6_1}) studied in Sections~\ref{sec6} and \ref{sec7} for $p = 10.3$. 
The parameter $p$ will be changed in the range $10.3 \le p \le 10.8$.

For numerical simulations we use regularization (\ref{eq6_2}) with random variable $\alpha_t$ uniformly distributed in the interval $[0.4,\, 0.5]$; this is the first type of noise denoted as $\mu$. For the statistical analysis we compute $10^5$ samples of the solution $u(1)$ for the same initial condition (\ref{eq4_1IC}) and different $N$  and $p$; note that we have reduced the number of samples since we need to achieve larger $N$. Figure~\ref{fig11} shows the dependence of the expectation value $\mathbb{E}[u_2(1)]$ and the standard deviation $\sigma[u_2(1)]$ as functions of the parameter $p$ for the two sequential viscous scales, $N = 19$ (crosses) and $20$ (circles).

\begin{figure}[t]
\centering
\includegraphics[width=0.9\textwidth]{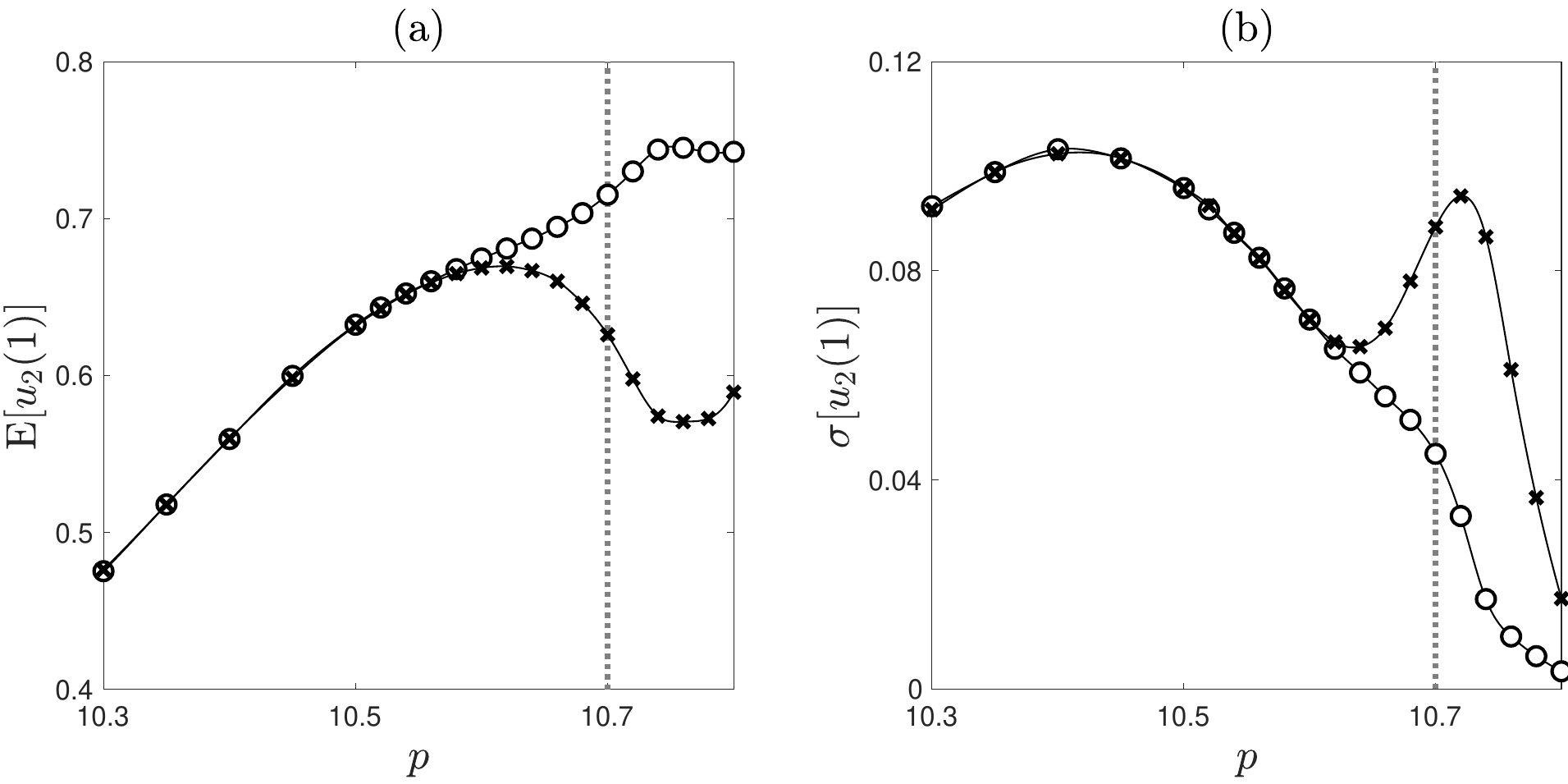}
\caption{(a) Expectation value $\mathbb{E}[u_2(1)]$ and (b) the standard deviation $\sigma[u_2(1)]$ as functions of the parameter $p$ for the two sequential viscous scales, $N = 19$ (crosses) and $20$ (circles). All markers are measurement points. 
The vertical dashed line corresponds to $p = 10.7$; see Fig.~\ref{fig12}.
The period-doubling bifurcation can be recognized at $p \approx 10.65$.}
\label{fig11}
\end{figure}

We observe that the spontaneously stochastic RG attractor bifurcates around $p \approx 10.65$. 
For larger $p$ the attractor becomes periodic with period $2$, which means that we are dealing with a period-doubling bifurcation.
We verify the periodicity of the RG attractor in Fig.~\ref{fig12}, where we plot the PDFs of the components $u_2(1)$, $u_3(1)$ and $u_4(1)$ for $p = 10.7$ and different $N$. 
The plots confirm that the regularized stochastic solutions converge in the inviscid limit for even and odd subsequences, $N = 2M \to \infty$ and $N = 2M+1 \to \infty$, and the limits are given by two different probability distributions. 
For flow kernels, this means
	\begin{equation}
	\label{eq8_2}
	\Phi_a^\infty = \lim_{M \to \infty}\Phi^{(2M)}, \quad
	\Phi_b^\infty = \lim_{M \to \infty}\Phi^{(2M+1)}, \quad
	\Phi^{(0)} \in \mathcal{B}(\Phi_a^\infty,\Phi_b^\infty),
	\end{equation}
in analogy with the deterministic case (\ref{eq5_2}). Note that the two limits $\Phi_a^\infty$ and $\Phi_b^\infty$ in Eq.~(\ref{eq8_2}) are defined up to a permutation, depending on the initial flow kernel $\Phi^{(0)}$. For example, these limits are swapped for the second type of noise denoted by $\tilde\mu$, for which $\alpha_t$ is uniformly distributed in the interval $[0.3,\,0.301]$; see the dotted lines in Fig.~\ref{fig12}. 
We do not have enough statistics to reliably measure the RG eigenvalue as a function of $p$. 
Recall, however, that $\rho \approx -0.7$ at $p = 10.3$, which is close to the value $-1$ required for the period doubling~\cite{guckenheimer2013nonlinear}. 

\begin{figure}[tp]
\centering
\includegraphics[width=0.99\textwidth]{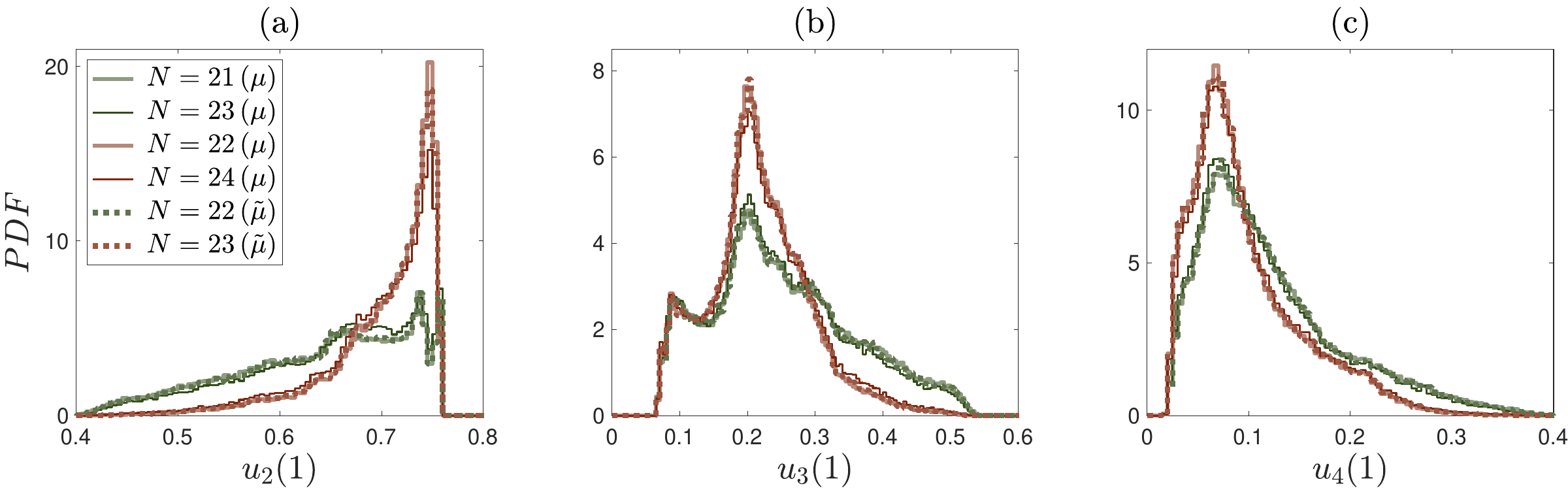}
\caption{PDFs of the state components (a) $u_2(1)$, (b) $u_3(1)$ and (c) $u_4(1)$ computed with $10^5$ simulations for $p = 10.7$. The graphs correspond to different $N = 21,\ldots,24$ and two types of noise, $\mu$ and $\tilde\mu$. For each type of noise, PDFs converge to two different statistics (green and red) depending on the parity of $N$.}
\label{fig12}
\end{figure}

Our findings suggest that the existence and universality of spontaneously stochastic solutions are explained by the existence of stochastic RG attractors. The properties of these attractors are described using classical stability and bifurcation theories. In particular, our results show that there can be different types of spontaneous stochasticity depending on the type of the corresponding RG attractor.

\section{Conclusion}
\label{sec10}

We designed a class of models for studying the limit of vanishing regularization and noise (inviscid limit) in multiscale systems. For these models we presented further development of the RG approach, originally introduced in~\cite{mailybaev2023spontaneous}, which relates the inviscid limit to the dynamics of the RG operator. One can draw two types of conclusions from our analysis. First, the RG approach explains known and reveals new universal properties of solutions obtained in the inviscid limit. 
Here, universality (independence from regularization and noise) follows from the correspondence of the inviscid limit to the RG attractor.
In particular, this explains the existence and universality of spontaneously stochastic solutions, i.e., solutions that remain stochastic in the limit of vanishing noise. A new property is the asymptotic universality of the first-order terms (deviations) obtained from the leading RG eigenmode of the linearized RG operator. 

The second type of conclusions relates to bifurcations of RG attractors with a change of the ideal system.
Classical bifurcation theory applies here, as shown by several examples of period-doubling bifurcations.
This and potentially other types of bifurcations point to another interesting property of the inviscid limit: the variety of forms it can take and the connections between them.
In particular, we observed a period-doubling RG bifurcation for spontaneously stochastic solutions. 
It leads to two different (but still universal) statistical behaviors for the inviscid limits coexisting in the same system.

Our models are discrete in both space and time. 
Compared to previous studies~\cite{mailybaev2023spontaneously,mailybaev2023spontaneous}, we introduced an energy balance to the equations of motion that reproduces some distinctive properties of turbulence, such as the energy cascade and intermittency (see the Appendix). 
Thus, our models can serve as prototypes for the limits of vanishing viscosity and noise in the Navier--Stokes and Burgers systems, among many others. 
Let us emphasize that the fractal lattice model has significant advantages.
From a theoretical point of view, it allows an explicit formulation of the RG operator. 
From a numerical point of view, it significantly increases the speed and accuracy of simulations.
Note that some of the RG constructions of this paper have already been generalized and verified on shell models of turbulence in which time is continuous~\cite{mailybaev2024rg}. 

The developed RG approach suggests several directions for future research. 
Already in a few simple examples of this paper we observed a wide range of scenarios such as fixed-point, periodic and chaotic RG attractors. 
One can further investigate the RG dynamics in its relation to the classical theory of dynamical systems. 
This study can be carried out both numerically and analytically, since simple solvable systems with nontrivial properties can be constructed on the fractal lattice~\cite{mailybaev2023spontaneously,mailybaev2023spontaneous}. 
Another promising direction is the generalization and analysis of the RG features in physical models.
Understanding the underlying RG dynamics can greatly facilitate this research by providing a variety of feasible tests.

\section{Appendix: Energy cascade and intermittency}
\label{sec11}

In this section we describe the statistical steady state in a forced model. The forcing (energy input) is applied at the largest scale of our model by modifying Eq.~(\ref{eq2_4})  as
	\begin{equation}
	\label{eqA_1}
	u_0(t+1) = \big[1-f_0(t)\big] u_0(t)+1, \quad t = 0,1,\ldots
	\end{equation}
Here unit energy is added at each large-scale turnover time $\tau_0 = 1$.
At smaller scales ($n \ge 1$), let us consider the regularized model given by Eqs.~(\ref{eq2_2}), (\ref{eq2_7}), (\ref{eq2_6}) and  (\ref{eq6_1}). We introduce the structure functions of order $p$ as
	\begin{equation}
	\label{eqA_2}
	S_p(\ell_n) = \left\langle \left[u_n(t)\right]^p \right\rangle 
	= \lim_{M \to \infty} \frac{1}{M}\sum_{m = 0}^{M-1} \left[u_n(m\tau_n)\right]^p,
	\end{equation}
where averaging at each scale is performed over discrete times $t = m\tau_n$ of the lattice.
The numerical calculation of these structure functions was carried out for $N = 17$ and $\alpha = 1/4$ starting from the initial condition (\ref{eq4_1IC}). We ignored the initial (transient) time interval of length $100$, and then averaged over a time interval of length $4 \times 10^4$. The averages converged accurately and the results are shown in Fig.~\ref{figApp1}(a) in log-log scale.

\begin{figure}[tp]
\centering
\includegraphics[width=0.9\textwidth]{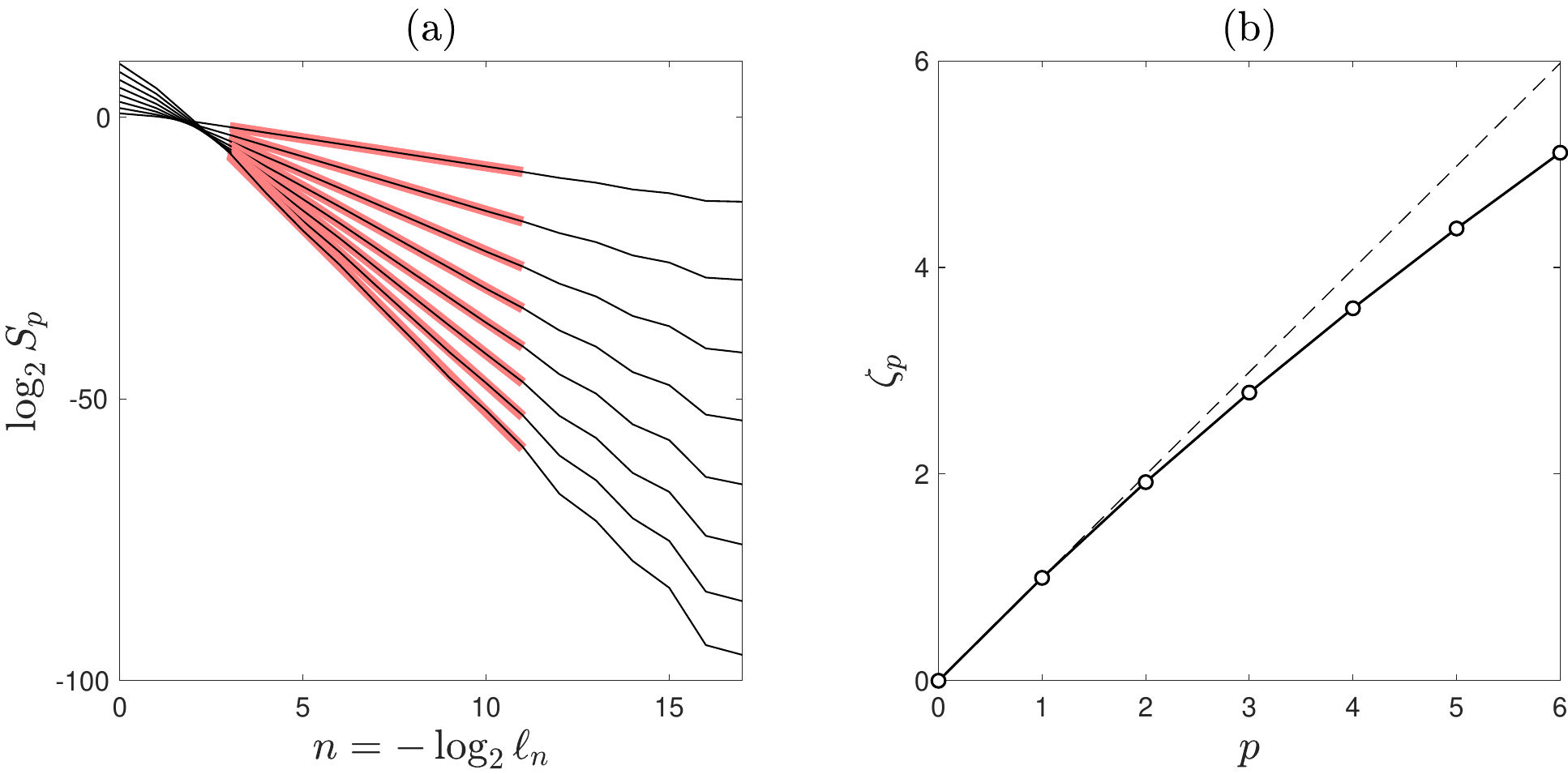}
\caption{(a) Structure functions for $p = 1,\ldots,8$. (b) Anomalous scaling exponents.}
\label{figApp1}
\end{figure}

Recall that $u_n(t)$ is interpreted as the energy at scale $n$ and time $t$. Hence, the first-oder structure function $S_1(\ell_n)$ represents the time-averaged energy at each scale. The existence of the statistical stationary state implies that all the injected energy, which has unit rate according to Eq.~(\ref{eqA_1}), is transported to the viscous scale $N$, where it is eventually dissipated by the regularization mechanism. 
In turbulence, this phenomenon is called the energy cascade from large to small scales~\cite{frisch1999turbulence,eyink2006onsager,eyink2024onsager}.

The structure functions in Fig.~\ref{figApp1}(a) have a power-law dependence
	\begin{equation}
	\label{eqA_3}
	S_p(\ell_n) \propto \ell_n^{\zeta_p}
	\end{equation}
in the inertial interval, as shown by the thick red lines. Here, the inertial interval represents a range of scales $\ell_n$ that are much smaller than the forcing scale $\ell_0$ and much larger than the viscous scale $\ell_N$. The exponents $\zeta_p$ presented in Fig.~\ref{figApp1}(b) depend nonlinearly on the order $p$; for a reference, we plotted a dashed straight line through the origin and the point $(1,\zeta_1)$. This nonlinear dependence of the (anomalous) exponents $\zeta_p$ is the well-known signature of small-scale intermittency~\cite{frisch1999turbulence}. 

Finally, we recall that the energy transferred from scale $\ell_n$ to $\ell_{n+1}$ in one turnover time $\tau_n = 2^{-n}$ is equal to $f_n(t)u_n$. Since $f_n(t) \sim 1$ by order of magnitude, we estimate the energy flux through the scale $\ell_n$ to be of order $u_n/\tau_n$. Therefore, the condition of constant mean energy flux yields $\langle u_n \rangle \propto \tau_n = \ell_n^{\zeta_1}$ with $\zeta_1 = 1$. This result agrees very accurately with numerical simulations and represents an analogue of the exponent $\zeta_3 = 1$ in three-dimensional Navier-Stokes turbulence.

\vspace{2mm}\noindent\textbf{Acknowledgments.} 
This work was supported by CNPq grant 308721/2021-7, FAPERJ grant E-26/201.054/2022 and CAPES grant AMSUD3169225P. 

\vspace{2mm}\noindent\textbf{Data Accessibility.} 
Data available on reasonable request.

\vspace{2mm}\noindent\textbf{Conflict of interest.}
The author declares that he has no potential conflict of interest related to this work.

\bibliography{refs}

\end{document}